%% file: main.tex
\documentclass[11pt]{article}
\usepackage{amssymb}

\usepackage{booktabs} 
\usepackage{multirow}
\usepackage[ruled]{algorithm2e} 

\SetAlFnt{\small}
\SetAlCapFnt{\small}
\SetAlCapNameFnt{\small}
\SetAlCapHSkip{0pt}
\IncMargin{-\parindent}

\usepackage[table]{xcolor}

\usepackage{bm}

\newcommand{\defeq}{\coloneqq}

\usepackage{complexity}
\newcommand{\CLS}{\mathsf{CLS}}
\newcommand{\NPP}{\mathsf{NP}}
\let\R\relax

\newcommand{\hvx}{\hat{\vx}}
\newcommand{\hvy}{\hat{\vy}}
\newcommand{\hvz}{\hat{\vz}}

\newcommand{\lineA}{\overline{\mat{A}}}

\usepackage[hidelinks, colorlinks = true, linkcolor=blue, citecolor=green]{hyperref}

\newcommand{\declarecolor}[2]{\definecolor{#1}{RGB}{#2}\expandafter\newcommand\csname #1\endcsname[1]{\textcolor{#1}{##1}}}
\declarecolor{White}{255, 255, 255}
\declarecolor{Black}{0, 0, 0}
\declarecolor{Maroon}{128, 0, 0}
\declarecolor{Coral}{255, 127, 80}
\declarecolor{Red}{182, 21, 21}
\declarecolor{LimeGreen}{50, 205, 50}
\declarecolor{DarkGreen}{0, 80, 0}
\declarecolor{Purple}{146, 42, 158}
\declarecolor{Navy}{0, 0, 128}
\declarecolor{LightBlue}{84, 101, 202}
\definecolor{mydarkblue}{rgb}{0,0.08,0.45}
\hypersetup{ %
    pdftitle={},
    pdfkeywords={},
    pdfborder=0 0 0,
    pdfpagemode=UseNone,
    colorlinks=true,
    linkcolor=Navy,
    citecolor=DarkGreen,
    filecolor=Purple,
    urlcolor=Black,
}

\usepackage{fullpage}
\usepackage{authblk}

\usepackage{MnSymbol}

\input{math-commands}

\usepackage{nicefrac}

\usepackage{tikz}
\usetikzlibrary{positioning}
\usetikzlibrary{shapes,arrows}

\usepackage{subfigure}

\usepackage[%
linewidth=2pt,
linecolor=gray,
middlelinecolor= black,
middlelinewidth=0.4pt,
roundcorner=1pt,
topline = false,
rightline = false,
bottomline = false,
rightmargin=0pt,
skipabove=0pt,
skipbelow=0pt,
leftmargin=0pt,
innerleftmargin=4pt,
innerrightmargin=0pt,
innertopmargin=0pt,
innerbottommargin=0pt,
]{mdframed}

\usepackage[capitalize,noabbrev,nameinlink]{cleveref}

\usepackage{enumitem}

\usepackage[natbib=true,maxcitenames=3,maxnames=3,style=alphabetic]{biblatex}
\usepackage{biblatex2bibitem}

\renewcommand{\vec}[1]{\bm{#1}}
\newcommand{\mat}[1]{\mathbf{#1}}

\newenvironment{nproblem}[1][\unskip]{%
  \medskip
  \begin{mdframed}
  \noindent
  \textbf{\underline{$#1$ Problem.}} \\
  \noindent
}{%
  \end{mdframed}
  \medskip 
}

\newcommand{\gdaFixed}{\textsc{GDAFixedPoint}}
\newcommand{\symgdaFixed}{\textsc{SymGDAFixedPoint}}
\newcommand{\nsymgdaFixed}{\textsc{NonSymGDAFixedPoint}}

\usepackage{thm-restate}

\theoremstyle{plain}
\newtheorem{theorem}{Theorem}[section]
\newtheorem{lemma}[theorem]{Lemma}
\newtheorem{corollary}[theorem]{Corollary}
\newtheorem{proposition}[theorem]{Proposition}

\newtheorem{assumption}[theorem]{Assumption}

\theoremstyle{definition}
\newtheorem{definition}[theorem]{Definition}

\theoremstyle{remark}
\newtheorem{remark}[theorem]{Remark}

\title{The Complexity of Symmetric Equilibria in Min-Max Optimization and Team Zero-Sum Games\thanks{The authors are ordered alphabetically.}}

\author[1]{Ioannis Anagnostides}
\author[2]{Ioannis Panageas}
\author[3]{Tuomas Sandholm}
\author[4]{Jingming Yan}

\affil[1,3]{Carnegie Mellon University}
\affil[2,4]{University of California, Irvine}
\affil[3]{Strategy Robot, Inc.}
\affil[3]{Strategic Machine, Inc.}
\affil[3]{Optimized Markets, Inc.}
\affil[ ]{\texttt{\{ianagnos,sandholm\}}\texttt{@cs.cmu.edu}, \texttt{\{ipanagea,jingmy1\}}\texttt{@uci.edu}}

\bibliography{main}

\begin{document}

\maketitle

\begin{abstract}
    We consider the problem of computing stationary points in min-max optimization, with a particular focus on the special case of computing Nash equilibria in (two-)team zero-sum games.

    We first show that computing $\epsilon$-Nash equilibria in $3$-player \emph{adversarial} team games---wherein a team of $2$ players competes against a \emph{single} adversary---is \textsf{CLS}-complete, resolving the complexity of Nash equilibria in such settings. Our proof proceeds by reducing from \emph{symmetric} $\epsilon$-Nash equilibria in \emph{symmetric}, identical-payoff, two-player games, by suitably leveraging the adversarial player so as to enforce symmetry---without disturbing the structure of the game. In particular, the class of instances we construct comprises solely polymatrix games, thereby also settling a question left open by Hollender, Maystre, and Nagarajan (2024). We also provide some further results concerning equilibrium computation in adversarial team games.

    Moreover, we establish that computing \emph{symmetric} (first-order) equilibria in \emph{symmetric} min-max optimization is \textsf{PPAD}-complete, even for quadratic functions. Building on this reduction, we further show that computing symmetric $\epsilon$-Nash equilibria in symmetric, $6$-player ($3$ vs. $3$) team zero-sum games is also \textsf{PPAD}-complete, even for $\epsilon = \text{poly}(1/n)$. As an immediate corollary, this precludes the existence of symmetric dynamics---which includes many of the algorithms considered in the literature---converging to stationary points. Finally, we prove that computing a \emph{non-symmetric} $\text{poly}(1/n)$-equilibrium in symmetric min-max optimization is \FNP-hard.
\end{abstract}

\pagenumbering{gobble}

\clearpage
\tableofcontents

\clearpage

\pagenumbering{arabic}

\input{text/intro}
\input{text/background}
\input{text/cls}
\input{text/ppadminmax}

\input{text/conclusions}

\section*{Acknowledgments}
I.P. would like to acknowledge ICS research award and a start-up grant from UCI. Part of this work was done while I.P. and J.Y. were visiting Archimedes Research Unit. 
T.S. is supported by the Vannevar Bush Faculty Fellowship ONR N00014-23-1-2876, National Science Foundation grants RI-2312342 and RI-1901403, ARO award W911NF2210266, and NIH award A240108S001. We are grateful to Alexandros Hollender for many valuable discussions.


\printbibliography

\appendix
\input{text/appendix}

\end{document}

%% file: math-commands.tex
\usepackage{multirow,array}
\usepackage{diagbox}

\newcommand{\matA}{\mat{A}}
\newcommand{\vecb}{\vec{b}}

\def\ve{{\bm{e}}}

\def\vw{{\bm{w}}}
\def\vx{{\bm{x}}}
\def\vy{{\bm{y}}}
\def\vz{{\bm{z}}}

\DeclareMathAlphabet{\mathsfit}{\encodingdefault}{\sfdefault}{m}{sl}
\SetMathAlphabet{\mathsfit}{bold}{\encodingdefault}{\sfdefault}{bx}{n}

\newcommand{\calA}{\ensuremath{\mathcal{A}}}

\newcommand{\calC}{\ensuremath{\mathcal{C}}}

\newcommand{\calX}{\ensuremath{\mathcal{X}}}
\newcommand{\calY}{\ensuremath{\mathcal{Y}}}

\newcommand{\vxstar}{\vx^*}
\newcommand{\vystar}{\vy^*}

\newcommand{\Amax}{\mat{A}_{\max}}
\newcommand{\Amin}{\mat{A}_{\min}}

\newcommand{\R}{\mathbb{R}}

\newcommand{\norm}[1]{\left\| #1 \right\|}

\DeclareMathOperator*{\argmin}{argmin}

\usepackage{amsthm}
\usepackage{mathtools}
\usepackage{amsmath}
\usepackage{bbm}
\usepackage{amsfonts}



%% file: text/intro.tex
\section{Introduction}

We consider the fundamental task of computing local equilibria in constrained min-max optimization problems of the form
\begin{equation}
    \label{eq:minimax}
    \min_{\vx \in \calX} \max_{\vy \in \calY} f(\vx, \vy),
\end{equation}
where $\calX \subseteq \R^{d_x}$ and $\calY \subseteq \R^{d_y}$ are convex and compact constraint sets, and $f : \calX \times \calY \to \R$ is a smooth objective function. Tracing all the way back to Von Neumann's celebrated minimax theorem~\cite{vonNeumann28:Zur} and the inception of game theory, such problems are attracting renewed interest in recent years propelled by a variety of modern machine learning applications, such as generative modeling~\citep{goodfellow2014generative}, reinforcement learning~\citep{daskalakis2020independent,Bai20:Provable,Wei21:Last}, and adversarial robustness~\citep{Madry17:Towards,cohen19:Certified,Bai21:Recent,Carlini19:Evaluating}. Another prominent class of problems encompassed by~\eqref{eq:minimax} concerns computing Nash equilibria in \emph{(two-)team zero-sum games}~\citep{Zhang23:Team,Zhang21:Computing,Basilico17:Team,Stengel97:Team,Carminati23:Hidden,Orzech23:Correlated,Farina18:Ex,Zhang20:Converging,Celli18:Computational,schulman2017duality}, which is a primary focus of this paper.

Perhaps the most natural solution concept---guaranteed to always exist---pertaining to~\eqref{eq:minimax}, when $f$ is nonconvex-nonconcave, is a pair of strategies $(\vx^*, \vy^*)$ such that both players (approximately) satisfy the associated \emph{first-order} optimality conditions~\citep{Tsaknakis21:Finding,Jordan23:First,Ostrovskii21:Efficient,Nouiehed19:Solving}, as formalized in the definition below.
\begin{definition}
    \label{def:FONE}
    A point $(\vxstar, \vystar) \in \calX \times \calY$ is an \emph{$\epsilon$-first-order Nash equilibrium} of~\eqref{eq:minimax} if
    \begin{equation*}
        \langle \vx - \vxstar, \nabla_{\vx} f(\vxstar, \vystar) \rangle \geq -\epsilon \quad \text{and} \quad \langle \vy - \vystar, \nabla_{\vy} f(\vxstar, \vystar) \rangle \leq  \epsilon \quad \forall (\vx, \vy) \in \calX \times \calY.
    \end{equation*}
\end{definition}
\Cref{def:FONE} can be equivalently recast as a variational inequality (VI) problem: if $\vz \defeq (\vx, \vy)$ and $F : \vz \mapsto F(\vz) \defeq (\nabla_{\vx} f(\vx, \vy), -\nabla_{\vy} f(\vx, \vy))$, we are searching for a point $\vz^* \in \mathcal{Z} \defeq \calX \times \calY$ such that $\langle \vz - \vz^*, F(\vz^*) \rangle \geq - 2 \epsilon$ for all $\vz \in \mathcal{Z}$. Yet another equivalent definition is instead based on approximate \emph{fixed points} of gradient descent/ascent (GDA); namely, \Cref{def:FONE} amounts to bounding the gradient mappings
\begin{equation}
\label{eq:fpgda}
\tag{Fixed points of GDA} 
\| \vx^* - \Pi_{\calX} ( \vx^* - \nabla_{\vx} f(\vx^*, \vy^*) ) \|\leq \epsilon', \| \vy^* - \Pi_{\calY} ( \vy^* + \nabla_{\vy} f(\vx^*, \vy^*) ) \|  \leq \epsilon'
\end{equation} 
for some approximation parameter $\epsilon' > 0$ that is (polynomially) dependent on $\epsilon > 0$, where $\|\cdot\|$ is the (Euclidean) $\ell_2$ norm and $\Pi(\cdot)$ is the projection operator. Other definitions that differentiate between the order of play between players---based on the notion of a Stackelberg equilibrium---have also been considered in the literature~\citep{Jin20:What}.

The complexity of min-max optimization is well-understood in certain special cases, such as when $f$ is convex-concave (\emph{e.g.},~\citet{Korpelevich76:Extragradient,MertikopoulosLZ19,Cai22:Finite,Choudhury23:Single,Gorbunov22:Last}, and references therein), or more broadly, nonconvex-concave~\citep{Lin20:Gradient,Xu23:Unified,Stochastic20:Luo}. However, the complexity of general min-max optimization problems, when the objective function $f$ is nonconvex-nonconcave, has remained wide open despite intense efforts in recent years. \citet{DSZ21} made progress by establishing certain hardness results targeting the more challenging setting in which there is a \emph{joint} (that is, coupled) set of constraints. In fact, it turns out that their lower bounds apply even for linear-nonconcave objective functions (\emph{cf.}~\citet{Bernasconi24:Role}), showing that their hardness result is driven by the presence of joint constraints---indeed, under uncoupled constraints, many efficient algorithms attaining~\Cref{def:FONE} (for linear-nonconcave problems) have been documented in the literature. In the context of min-max optimization, the most well-studied setting posits that players have independent constraints; this is the primary focus of our paper.

\subsection{Our results}

We establish new complexity lower bounds in min-max optimization for computing equilibria in the sense of~\Cref{def:FONE}; our main results are gathered in~\Cref{tab:results}.

\begin{table}[h!]
\centering
\caption{The main results of this paper. NE stands for Nash equilibrium and FONE for first-order Nash equilibrium (\Cref{def:FONE}). We also abbreviate symmetric to ``sym.'' (second column).}
\footnotesize 
\renewcommand{\arraystretch}{1.2} 
\begin{tabular}{lccc}
\toprule
\textbf{Class of problems}  & \textbf{Eq. concept} & \textbf{Complexity} & \textbf{Even for} \\
\midrule
\rowcolor{gray!20} Adversarial team games & $\epsilon$-NE & $\CLS$-complete (\Cref{theorem:actual}) & $3$-player ($2$ vs. $1$), polymatrix \\ \multirow{2}{*}{Symmetric min-max} & sym. $\epsilon$-FONE & \PPAD-complete (\Cref{theorem:team-hard}) & polymatrix, team $0$-sum, $\epsilon = \nicefrac{1}{n^c}$ \\
& non-sym. $\epsilon$-FONE & \NP-hard (\Cref{theorem:non-symmetric}) & quadratics, $\epsilon = \nicefrac{1}{n^c}$ \\
\bottomrule
\end{tabular}
\label{tab:results}
\end{table}

\paragraph{Adversarial team games.} We first examine an important special case of~\eqref{eq:minimax}: \emph{adversarial team games}~\citep{Stengel97:Team}. Here, a team of $n$ players with identical interests is competing against a single adversarial player. (In such settings, \Cref{def:FONE} captures precisely the Nash equilibria of the game.) The computational complexity of this problem was placed by~\citet{Anagnostides23:Algorithms} in the complexity class~$\CLS$---which stands for continuous local search~\citep{Daskalakis11:Continuous}. Further, by virtue of a result of~\citet{Babichenko21:Settling}, computing Nash equilibria in adversarial team games when $n \gg 1$ is~$\CLS$-complete. In the context of this prior work, an important question left open by~\citet{Anagnostides23:Algorithms} concerns the case where $n$ is a small constant, a regime not captured by the hardness result of~\citet{Babichenko21:Settling} pertaining to identical-interest games---for, in such games, one can simply identify the strategy leading to the highest payoff, which is tractable when $n$ is small.

We show that even when $n = 2$, computing an $\epsilon$-Nash equilibrium in adversarial team games is $\CLS$-complete. (Of course, the case where $n=1$ amounts to two-player zero-sum games, well known to be in $\P$.)

\begin{restatable}{theorem}{maincls}
    \label{theorem:cls-completeness}
    Computing an $\epsilon$-Nash equilibrium in $3$-player (that is, $2$ vs. $1$) adversarial team games is $\CLS$-complete.
\end{restatable}

Coupled with earlier results, \Cref{theorem:cls-completeness} completely characterizes the complexity landscape for computing Nash equilibria in adversarial team games.

Our proof is based on a recent hardness result of~\citet{ghosh2024complexitysymmetricbimatrixgames} (\emph{cf.}~\citet{Tewolde25:Computing}), who proved that computing a \emph{symmetric} $\epsilon$-Nash equilibrium in a symmetric two-player game with identical payoffs is~$\CLS$-complete. The key idea in our reduction is that one can leverage the adversarial player so as to enforce symmetry between the team players, without affecting the equilibria of the original game; the basic gadget underpinning this reduction is explained and analyzed in~\Cref{sec:cls}.

Incidentally, our~$\CLS$-hardness reduction hinges on a \emph{polymatrix} adversarial team game, thereby addressing another open question left recently by~\citet{Hollender24:Complexity}.

\begin{theorem}
    \label{theorem:actual}
    \Cref{theorem:cls-completeness} holds even when one restricts to polymatrix, $3$-player adversarial team games.
\end{theorem}

We complement the above hardness result by showing that, even when $n = 2$, there is an adversarial team game with a unique Nash equilibrium supported on \emph{irrational numbers} (\Cref{theorem:irrational})---unlike~\Cref{theorem:cls-completeness}, \Cref{theorem:irrational} is based on a non-polymatrix game, for otherwise rational Nash equilibria are guaranteed to exist. This strengthens an observation of~\citet{Stengel97:Team}, who noted that a \emph{team-maxmin equilibrium}---a stronger notion than Nash equilibria---can be supported on irrational numbers. \Cref{theorem:irrational} brings to the fore the complexity of deciding whether an adversarial team game admits a unique Nash equilibrium; \Cref{theorem:uniqueATG} addresses that question for the version of the problem accounting even for $\poly(1/n)$-Nash equilibria.

\paragraph{Symmetric min-max optimization} As we have seen, symmetry plays a key role in the proof of~\Cref{theorem:cls-completeness}, but that result places no restrictions on whether the equilibrium is symmetric or not---this is indeed the crux of the argument. The next problem we consider concerns computing \emph{symmetric} equilibria in \emph{symmetric} min-max optimization problems, in the following natural sense.

\begin{definition}[Symmetric min-max optimization]
    \label{def:symmetric}
  A function $f:\calX \times \calY \to \R$ is called \emph{antisymmetric} if $\calX = \calY$ and 
  \begin{equation*}
      f(\vx,\vy) = -f(\vy,\vx) \quad \forall (\vx, \vy) \in \calX \times \calY.
  \end{equation*}
  Furthermore, a point $(\vx, \vy) \in \calX \times \calY$ is called \emph{symmetric} if $\vx = \vy$. The associated min-max optimization problem is called symmetric if the underlying function $f$ is \emph{antisymmetric}.\footnote{The nomenclature of this definition is consistent with the usual terminology in the context of (two-player) zero-sum games: a symmetric zero-sum game is one in which the the underlying game matrix $\mat{A}$ is antisymmetric (that is, skew-symmetric), so that $\langle \vx, \mat{A} \vy \rangle = - \langle \vy, \mat{A} \vx \rangle$ for all $(\vx, \vy)$. Symmetric zero-sum games are ubiquitous in the literature and in practical scenarios alike.}
\end{definition}
The study of symmetric equilibria has a long history in the development of game theory, propelled by Nash's pathbreaking PhD thesis~\citep{Nash50:Non} (\emph{cf.}~\citet{Gale51:Symmetric}), and has remained a popular research topic ever since~\citep{Tewolde25:Computing,Emmons22:For,Garg18:Existential,ghosh2024complexitysymmetricbimatrixgames,Mehta14:Constant}. It is not hard to see that symmetric min-max optimization problems, in the sense of~\Cref{def:symmetric}, always admit \emph{symmetric} first-order Nash equilibria (\Cref{lem:exists}). What is more, we show that computing such a symmetric equilibrium is in the complexity class \PPAD~\citep{Papadimitriou94:Complexity}; this is based on an argument of~\citet{Etessami10:On}, and complements~\citet{DSZ21}, who proved that the problem of computing approximate fixed points of gradient descent/ascent---which they refer to as \textsc{GDAFixedPoint}---lies in \PPAD. In a celebrated series of work, it was shown that $\PPAD$ captures the complexity of computing Nash equilibria in finite games~\citep{Daskalakis09:The,Chen09:Settling}. In this context, we establish that $\PPAD$ also characterizes the complexity of computing symmetric first-order Nash equilibria in symmetric min-max optimization problems:

\begin{theorem}
    \label{theorem:main}
    Computing a symmetric $\nicefrac{1}{n^c}$-approximate first-order Nash equilibrium in symmetric $n$-dimensional min-max optimization is \PPAD-complete for any constant $c > 0$.
\end{theorem}

Barring major complexity breakthroughs, \Cref{theorem:main} precludes the existence of algorithms with complexity polynomial in the dimension and $1/\epsilon$, where $\epsilon > 0$ measures the precision (per\\~\Cref{def:FONE}), under the symmetry constraint of~\Cref{def:symmetric}. This stands in contrast to (nonconvex) minimization problems, wherein gradient descent converges to stationary points at a rate of $\poly(1/\epsilon)$; even in the regime where $\epsilon = 1/\exp(n)$, computing a stationary point of a smooth function is in~$\CLS$~\citep{Daskalakis11:Continuous}, which is a subclass of \PPAD~\citep{Fearnley23:Complexity}. In fact, our reduction also rules out the existence of polynomial-time algorithms even when $\epsilon = \Theta(1)$ under some well-believed complexity assumptions (\Cref{cor:constant}).

The proof of~\Cref{theorem:main} is elementary, and is based on the \PPAD-hardness of computing \emph{symmetric} Nash equilibria in \emph{symmetric} two-player games. Importantly, our reduction gives an immediate, and significantly simpler, proof (\Cref{thm:simple}) of the \PPAD-hardness result of~\citet{DSZ21}, while being applicable even with respect to quadratic and anti-symmetric functions defined on a product of simplexes. (Independently, \citet{Bernasconi24:Role} also considerably simplified the proof of~\citet{DSZ21}, although our approach here is to a large extent different.) The basic idea of our proof is that one can enforce the symmetry constraint of~\Cref{theorem:main} via a coupled constraint set.

Furthermore, as a byproduct of~\Cref{theorem:main}, it follows that any \emph{symmetric} dynamics---whereby both players follow the same online algorithm, as formalized in~\Cref{def:sym-dynamics}---cannot converge to a first-order Nash equilibrium in polynomial time, subject to $\PPAD \neq \P$ (\Cref{theorem:sym-dyn}). This already captures many natural dynamics for which prior papers in the literature (\emph{e.g.},~\citet{kalogiannis2021teamwork}) have painstakingly shown lack of convergence; \Cref{theorem:sym-dyn} provides a complexity-theoretic justification for such prior results, while precluding at the same time a much broader family of algorithms.

\paragraph{The complexity of non-symmetric equilibria} Remaining on symmetric min-max optimization, one natural question arising from~\Cref{theorem:main} concerns the complexity of \emph{non-symmetric} equilibria---defined as having distance at least $\delta > 0$. Unlike their symmetric counterparts, non-symmetric first-order Nash equilibria are not guaranteed to exist. In fact, we establish the following result.

\begin{restatable}{theorem}{nonsymmetric}
    \label{theorem:non-symmetric}
    For a symmetric min-max optimization problem, constants $c_1, c_2 > 0$, and $\epsilon = n^{-c_1}$, it is \NP-hard to distinguish between the following two cases:
    \begin{itemize}[noitemsep,topsep=0pt]
        \item any $\epsilon$-first-order Nash equilibrium $(\vx^*, \vy^*)$ satisfies $ \|\vx^* - \vy^* \| \leq n^{-c_2}$, and
        \item there is an $\epsilon$-first-order Nash equilibrium $(\vx^*, \vy^*)$ such that $ \|\vx^* - \vy^*\| \geq \Omega(1)$.
    \end{itemize}
\end{restatable}

\noindent The main technical piece, which is also the basis for~\Cref{theorem:uniqueATG} discussed earlier, is~\Cref{theorem:symmetric-new}, which concerns symmetric, identical-interest, two-player games. It significantly refines the hardness result of~\citet{MCLENNAN2010683} by accounting even for $\poly(1/n)$-Nash equilibria.

\paragraph{Team zero-sum games} Finally, building on the reduction of~\Cref{theorem:main} coupled with the gadget behind~\Cref{theorem:cls-completeness}, we establish similar complexity results for team zero-sum games, which generalize adversarial team games by allowing the presence of multiple adversaries. 
In particular, a \emph{symmetric} two-team zero-sum game and a \emph{symmetric} equilibrium thereof are in accordance with~\Cref{def:symmetric}---no symmetry constraints are imposed within the same team, but only across teams. We obtain a result significantly refining~\Cref{theorem:main}.

\begin{restatable}{theorem}{teamhard}
    \label{theorem:team-hard}
    Computing a symmetric $\nicefrac{1}{n^c}$-Nash equilibrium in symmetric, $6$-player ($3$ vs. $3$) team zero-sum polymatrix games is \PPAD-complete for some constant $c>0$.
\end{restatable}

Unlike our reduction in~\Cref{theorem:main} that comprises quadratic terms, the crux in team zero-sum games is that one needs to employ solely multilinear terms. The basic idea is to again use the gadget underpinning~\Cref{theorem:cls-completeness}, which enforces symmetry without affecting the equilibria of the game, thereby (approximately) reproducing the objective function that establishes~\Cref{theorem:main}.

It is interesting to note that the class of polymatrix games we construct to prove~\Cref{theorem:team-hard} belongs to a certain family introduced by~\citet{Cai11:Minmax}: one can partition the players into $2$ groups, so that any pairwise interaction between players of the same group is a coordination game, whereas any pairwise interaction across groups is a zero-sum game. \citet{Cai11:Minmax} showed that computing a Nash equilibrium is \PPAD-hard in the more general case where there are $3$ groups of players. While the complexity of that problem under $2$ groups remains wide open, \Cref{theorem:team-hard} shows \PPAD-hardness for computing \emph{symmetric} Nash equilibria in such games.

Taken together, our results bring us closer to characterizing the complexity of computing equilibria in min-max optimization.

\subsection{Further related work}

Adversarial team games have been the subject of much research tracing back to the influential work of~\citet{Stengel97:Team}, who introduced the concept of a \emph{team maxmin equilibrium (TME)}; a TME can be viewed as the best Nash equilibrium for the team. Notwithstanding its intrinsic appeal, it turns out that computing a TME is \FNP-hard~\citep{Borgs10:The}. Indeed, unlike two-player zero-sum games, team zero-sum games generally exhibit a \emph{duality gap}---characterized in the work of~\citet{schulman2017duality}.

This realization has shifted the focus of contemporary research to exploring more permissive solution concepts. One popular such relaxation is \emph{TMECor}, which enables team players to \emph{ex ante} correlate their strategies~\citep{Zhang23:Team,Zhang21:Computing,Basilico17:Team,Carminati22:A,Farina18:Ex,Zhang20:Converging,Celli18:Computational}. Yet, in the context of extensive-form games, computing a TMECor remains intractable; \citet{Zhang23:Team} provided an exact characterization of its complexity. Team zero-sum games can be thought of as two-player zero-sum games but with \emph{imperfect recall}, and many natural problems immediately become hard without perfect recall (\emph{e.g.}, \citet{Tewolde23:Computational}). Parameterized algorithms have been developed for computing a TMECor based on some natural measure of shared information~\citep{Zhang23:Team,Carminati22:A}. Beyond adversarial team games, \citet{Carminati23:Hidden} recently explored \emph{hidden-role} games, wherein there is uncertainty regarding which players belong in the same team, a feature that often manifests itself in popular recreational games---and used certain cryptographic primitives to solve them.

In contrast, this paper focuses on the usual Nash equilibrium concept, being thereby orthogonal to the above line of work. One drawback of Nash equilibria in adversarial team games is that the (worst-case) value of the team can be significantly smaller compared to TME~\citep{Basilico17:Team}. On the other hand, \citet{Anagnostides23:Algorithms} showed that $\epsilon$-Nash equilibria in adversarial team games admit an $\FPTAS$, which stands in stark contrast to TME, and indeed, Nash equilibria in general games~\citep{Daskalakis09:The,Chen09:Settling}. This was further strengthened by~\citet{kalogiannis2022efficiently,Kalogiannis24:Learning} for computing $\epsilon$-Nash equilibria in adversarial team \emph{Markov} games---the natural generalization to Markov (aka. stochastic) games.

Related to~\Cref{def:FONE} is the natural notion of a \emph{local} min-max equilibrium~\citep{DP18, DSZ21}. It is easy to see that any local min-max equilibrium---with respect to a sufficiently large neighborhood of $(\vx^*, \vy^*)$---must satisfy~\Cref{def:FONE}~\citep{DSZ21}. Unlike first-order Nash equilibria, local min-max equilibria are not guaranteed to exist.

Finally, another related work we should mention is \cite{mehta2015settling}, where it was shown that in two-player symmetric games, deciding whether a non-symmetric Nash equilibrium exists or not is \NP-hard. 

%% file: text/background.tex
\section{Preliminaries}

\paragraph{Notation} We use boldface lowercase letters, such as $\vx, \vy, \vz$, to represent vectors, and boldface capital letters, such as $\mat{A}, \mat{C}$, for matrices. We denote by $x_i$ the $i$th coordinate of a vector $\vx \in \R^n$. We use the shorthand notation $[n] \defeq \{1, 2, \dots, n\}$. $\Delta^n \defeq \{ \vx \in \R^n_{\geq 0} : \sum_{i=1}^n x_i = 1 \}$ is the probability simplex on $\R^n$. For $i \in [n]$, $\vec{e}_i \in \Delta^n$ is the $i$th unit vector. $\langle \cdot, \cdot \rangle$ denotes the inner product. For a vector $\vx \in \R^n$, $\|\vx\|_2 = \sqrt{\langle \vx, \vx \rangle}$ is its Euclidean norm. For $m \leq n$, $\vx_{[1 \cdots m]} \in \R^m$ is the vector containing the first $m$ coordinates of $\vx$. We sometimes use the $O(\cdot), \Theta(\cdot), \Omega(\cdot)$ notation to suppress absolute constants. We say that a continuously differentiable function $f$ is \emph{$L$-smooth} if its gradient is $L$-Lipschitz continuous with respect to $\|\cdot\|_2$; that is, $\| \nabla f(\vx) - \nabla f(\vx')  \|_2 \leq L \| \vx - \vx' \|_2$ for all $\vx, \vx'$ in the domain.

\paragraph{Two-player games} In a two-player game, represented in normal-form game, each player has a finite set, let $[n]$, of actions. Under a pair of actions $(i, j) \in [n] \times [n]$, the utility of the \emph{row} player is given by $\mat{R}_{i, j} $, where $\mat{R} \in \mathbb{Q}^{n \times n}$ is the payoff matrix of the row player. Further, we let $\mat{C} \in \mathbb{Q}^{n \times n}$ be the payoff matrix of the column player. Players are allowed to randomize by selecting mixed strategies---points in $\Delta^n$. Under a pair of mixed strategies $(\vx, \vy) \in \Delta^n \times \Delta^n$, the expected utility of the players is given by $\langle \vx, \mat{R} \vy \rangle$ and $\langle \vx, \mat{C} \vy \rangle$, respectively. The canonical solution concept in such games is the \emph{Nash equilibrium}~\citep{nash1951non}, which is recalled below.

\begin{definition}
    \label{def:NE}
    A pair of strategies $(\vx^*, \vy^*)$ is an \emph{$\epsilon$-Nash equilibrium} of $(\mat{R}, \mat{C})$ if
    \[
        \langle \vx^*, \mat{R} \vy^* \rangle \geq \langle \vx, \mat{R} \vy^* \rangle - \epsilon \quad \text{and} \quad \langle \vx^*, \mat{C} \vy^* \rangle \geq \langle \vx^*, \mat{C} \vy \rangle - \epsilon \quad \forall (\vx, \vy) \in \Delta^n \times \Delta^n.
    \]
\end{definition}

\paragraph{Symmetric two-player games} One of our reductions is based on \emph{symmetric} two-player games, meaning that $\mat{R} = \mat{C}^\top$. A basic fact, which goes back to~\citet{nash1951non}, is that any symmetric game admits a \emph{symmetric} Nash equilibrium $(\vx^*, \vx^*)$. Further, computing a Nash equilibrium in a general game can be readily reduced to computing a symmetric Nash equilibrium in a symmetric game~\citep[Theorem 2.4]{Nisan07:Algorithmic}. In conjunction with the hardness result of~\citet{Chen09:Settling}, we state the following immediate consequence.

\begin{theorem}[\citealp{Chen09:Settling}] \label{theorem:PPAD_for _symmetric}
    Computing a symmetric $\nicefrac{1}{n^c}$-Nash equilibrium in a symmetric two-player game is \PPAD-hard for any constant $c > 0$.
\end{theorem}

\paragraph{Team zero-sum games} A (two-)team zero-sum game is a multi-player game---represented in normal form for the purposes of this paper---in which the players' utilities have a certain structure; namely, we can partition the players into two (disjoint) subsets, such that each player within the same team shares the same utility, whereas players in different teams have opposite utilities---under any possible combination of actions. An \emph{adversarial} team game is a specific type of team zero-sum game wherein one team consists of a single player. As in~\Cref{def:NE} for two-player games, an \emph{$\epsilon$-Nash equilibrium} is a tuple of strategies such that no unilateral deviation yields more than an $\epsilon$ additive improvement in the utility of the deviator.

%% file: text/cls.tex
\section{Complexity of adversarial team games}

We begin by examining equilibrium computation in adversarial team games.

\subsection{$\CLS$-completeness for 3-player games}
\label{sec:cls}

Computing $\epsilon$-Nash equilibria in adversarial team games was placed in $\CLS$ by~\citet{Anagnostides23:Algorithms}, but whether $\CLS$ tightly characterizes the complexity of that problem remained open---that was only known when the number of players is large, so that the hardness result of~\citet{Babichenko21:Settling} can kick in. Our reduction here answers this question in the affirmative. 

We rely on a recent hardness result of~\citet{ghosh2024complexitysymmetricbimatrixgames} concerning symmetric, two-player games with identical payoffs. We summarize their main result below.

\begin{theorem}[\citealp{ghosh2024complexitysymmetricbimatrixgames}]
    \label{theorem:hard-2cls}
    Computing an $\epsilon$-Nash equilibrium in a symmetric, identical-payoffs, two-player game is~$\CLS$-complete.
\end{theorem}

Now, let $\mat{A} \in \mathbb{Q}^{n \times n}$ be the common payoff matrix of a two-player game, which satisfies $\mat{A} = \mat{A}^\top$ so that the game is symmetric. Without loss of generality, we will assume that $\mat{A}_{i, j} \leq -1$ for all $i, j \in [n]$. We denote by $\Amin$ and $\Amax$ the minimum and maximum entry of $\mat{A}$, respectively (which satisfy $\Amax, \Amin \leq -1$). The basic idea of our proof is to suitably use the adversarial player so as to force the other two players to play roughly the same strategy (\Cref{lemma:closeness}), while (approximately) maintaining the structure of the game (\Cref{lemma:small_z}).

\paragraph{Definition of the adversarial team game} Based on $\mat{A}$, we construct a $3$-player adversarial team game as follows. The utility function of the adversary reads
\begin{equation}
    \label{eq:util-atg}
    u(\vx, \vy, \vz) \defeq  \langle \vx, \mat{A} \vy \rangle  + \frac{|\Amin|}{\epsilon} \sum_{i = 1}^n \left(z_i(x_i - y_i) + z_{n + i} (y_i - x_i)\right) + z_{2n+1} |\Amin|.
\end{equation}
The adversary selects a strategy $ \vz \in \Delta^{2n + 1}$, while the team players, who endeavor to minimize~\eqref{eq:util-atg}, select strategies $\vx \in \Delta^n$ and $\vy \in \Delta^n$, respectively. (While the range of the utilities in~\eqref{eq:util-atg} grows with $\nicefrac{1}{\epsilon}$, normalizing to $[-1, 1]$ maintains all of the consequences by suitably adjusting the approximation.)

We begin by stating a simple auxiliary lemma (\Cref{sec:proofs1} contains the proof).

\begin{restatable}{lemma}{wellsuppo}
    \label{lemma:mass_on_small_reward}
    Let $(\vx_i^*, \vx^*_{-i})$ be an $\epsilon^2$-Nash equilibrium of a normal-form game, and $a_j$ any action of player $i$. If $u_i(a_k, \vx^*_{-i}) \leq u_i(a_j, \vx^*_{-i}) - c$ for some $c>0$ and $k \in [n]$, then $x^*_i(a_k) \leq \nicefrac{\epsilon^2}{c}.$
\end{restatable}

In a so-called \emph{well-supported} Nash equilibrium, every strategy $k$ that satisfies the precondition of~\Cref{lemma:mass_on_small_reward} would be played with probability $0$, but it will be easy to account for that slack.

The first important lemma establishes that, in equilibrium, $\vx \approx \vy$. The basic argument proceeds as follows. By construction of the adversarial team game~\eqref{eq:util-atg}, the adversary would be able to secure a large payoff whenever there is a coordinate $i \in [n]$ such that $|x_i - y_i| \gg 0$---by virtue of the second term in~\eqref{eq:util-atg}. But that cannot happen in equilibrium, for Player $\vx$ (or symmetrically Player $\vy$) can simply neutralize that term in the adversary's utility by playing $\vx = \vy$. We provide the precise, quantitative argument below.

\begin{lemma}[Equilibrium forces symmetry] \label{lemma:closeness}
    Consider an $\epsilon^2$-Nash equilibrium $(\vx^*, \vy^*, \vz^*)$ of the adversarial team game~\eqref{eq:util-atg} with $\epsilon^2 \leq \nicefrac{1}{2}$. Then, $\norm{\vx^* - \vy^*}_{\infty} \leq 2\epsilon$.
\end{lemma}

\begin{proof}
    For the sake of contradiction, suppose that $x_i^* - y_i^* > 2 \epsilon$ for some $i \in [n]$ (the case where $y_i^* - x_i^* > 2 \epsilon$ is symmetric, and can be treated analogously). Player $\vz$ could then choose action $a_i$ (with probability $1$), which secures a utility of
    \begin{align*}
        u(\vx^*, \vy^*, a_i) 
        & = \langle \vx^*, \mat{A} \vy^* \rangle  + \frac{|\Amin|}{\epsilon} \cdot (x^*_i - y^*_i) > \langle \vx^*, \mat{A} \vy^* \rangle + 2 |\Amin|,
    \end{align*}
    since $x_i^* - y_i^* > 2 \epsilon$. At the same time, Player $\vz$ could choose action $a_{2n+1}$, which secures a utility of $u(\vx^*, \vy^*, a_{2n+1}) = \langle \vx^*, \mat{A} \vy^* \rangle + |\Amin|$. So,
    \begin{equation*}
        u(\vx^*, \vy^*, a_i) - u(\vx^*, \vy^*, a_{2n+1}) \geq |\Amin|.
    \end{equation*}
    Applying~\Cref{lemma:mass_on_small_reward},
    \begin{align}
        \label{ineq:z}
    z^*_{2n+1} \leq \frac{\epsilon^2}{|\Amin|} \leq \epsilon^2.
    \end{align}
    Also, using the fact that $(\vx^*, \vy^*, \vz^*)$ is an $\epsilon^2$-Nash equilibrium,
    \begin{align}
    u(\vx^*,\vy^*,\vz^*) &\geq u(\vx^*,\vy^*,a_i) - \epsilon^2 \notag \\
    &\geq \langle \vx^*, \mat{A} \vy^* \rangle + 2|\Amin| -\epsilon^2 \notag \\
    &\geq \Amin + 2|\Amin|-\epsilon^2
    \notag \\
    &= |\Amin| -\epsilon^2, \label{ineq:utility}
    \end{align}
    since we have assumed that $\Amin < 0$. Now, consider the deviation of Player $\vx$ (from strategy $\vx^*$) to $\vx' \defeq \vy^*$. Then, $u(\vx', \vy^*, \vz^*) = \langle \vy^*, \mat{A} \vy^* \rangle + z^*_{2n+1} \cdot |\Amin|$. Thus, combining with~\eqref{ineq:utility} and~\eqref{ineq:z},
    \begin{align}
        u(\vx', \vy^*, \vz^*) - u(\vx^*, \vy^*, \vz^*) & \leq \langle \vy^*, \mat{A} \vy^* \rangle  + z_{2n+1}^* \cdot |\Amin| - |\Amin| + \epsilon^2 \notag \\
        & \leq \langle \vy^*, \mat{A} \vy^* \rangle + \epsilon^2 \cdot |\Amin| -  |\Amin | + \epsilon^2 \notag \\
        & \leq \Amax - |\Amin| + \epsilon^2(|\Amin|+1) \notag \\&\leq -1 < - \epsilon^2,  \label{eq:contract}
    \end{align}
    where we used that $\Amax, \Amin \leq -1$ and $\epsilon^2 \leq \nicefrac{1}{2}$. But~\eqref{eq:contract} contradicts the fact that $(\vx^*, \vy^*, \vz^*)$ is an $\epsilon^2$-Nash equilibrium since deviating to $\vx'$ yields a utility improvement (equivalently, decrease in cost) strictly larger than $\epsilon^2$. This completes the proof.
\end{proof}

Having established that $\vx \approx \vy$, the next step is to make sure that the adversarial player does not distort the original game by much. In particular, we need to make sure that the effect of the second term in~\eqref{eq:util-atg} is negligible. We do so by showing that $\vz_{2n+1} \approx 1$ (\Cref{lemma:small_z}).

The argument here is more subtle; roughly speaking, it goes as follows. Suppose that $z_{i} \gg 0$ or $z_{n+i} \gg 0$ for some $i \in [n]$. Since Player $\vz$ is approximately best responding, it would then follow that $|y_i^* - x_i^*| \gg 0$---otherwise Player $\vz$ would prefer to switch to action $2n+1$. But, if $|y_i^* - x^*_i| \gg 0$, Player $\vx$ could profitably deviate by reallocating probability mass by either removing or adding to $i$ (depending on whether $y_i^* - x_i^* > 0$), which leads to a contradiction.

\begin{lemma}[Most probability mass in $a_{2n+1}$] \label{lemma:small_z}
    Given any $\epsilon^2$-Nash equilibrium $(\vx^*, \vy^*, \vz^*)$ of the adversarial team game~\eqref{eq:util-atg} with $\epsilon \leq \nicefrac{1}{10}$, $z_j \leq 9 \epsilon$ for all $j \in [2n]$. In particular, $z_{2n+1} \geq 1 - 18 n \epsilon$.
\end{lemma}

\begin{proof}
    \Cref{lemma:closeness} shows that $\norm{\vx^* - \vy^*}_{\infty} \leq 2 \epsilon$. Let $i \in [n]$. We assume that $i$ is such that $x_i^* - y_i^* \geq 0$; the contrary case is symmetric. We consider two cases. First, suppose that $| x_i^* - y_i^*| \leq \nicefrac{\epsilon}{2}$. Then, we have
    \begin{align*}
        u(\vx^*, \vy^*, a_{2n+1}) - u(\vx^*, \vy^*, a_i) & \geq |\Amin| - \frac{|\Amin|}{\epsilon} ( x^*_i - y^*_i ) \\
        & \geq \frac{1}{2} |\Amin| \geq \frac{1}{2}.
    \end{align*}
    Thus, by~\cref{lemma:mass_on_small_reward}, we conclude that $z_i \leq 2 \epsilon^2$. Similarly,
    \begin{equation*}
        u(\vx^*, \vy^*, a_{2n+1}) - u(\vx^*, \vy^*, a_{n+i}) \geq |\Amin| - \frac{|\Amin|}{\epsilon} (  y^*_i - x_i^* ) \geq |\Amin| \geq 1,
    \end{equation*}
    since $x_i^* - y_i^* \geq 0$. Again, \Cref{lemma:mass_on_small_reward} implies that $z_{n+i} \leq \epsilon^2$.
    
    It thus suffices to treat the case where $| x_i^* - y_i^* | > \nicefrac{\epsilon}{2}$ (assuming that $x_i^* - y_i^* \geq 0$). It follows that there exists $j \in [n]$ such that $x^*_j - y^*_j < 0$. In addition, we observe that $u(\vx^*, \vy^*, a_{2n+1}) = \langle \vx^*, \mat{A} \vy^* \rangle + |\Amin| \geq \langle \vx^*, \mat{A} \vy^* \rangle + 1$, whereas $u(\vx^*, \vy^*, a_j) < \langle \vx^*, \mat{A} \vy^* \rangle$ and $u(\vx^*, \vy^*, a_{n+i}) < \langle \vx^*, \mat{A} \vy^* \rangle$. As a result, \Cref{lemma:mass_on_small_reward} implies that $z^*_{n+i} \leq \epsilon^2$ and $z^*_{j} \leq \epsilon^2$. 
    
    Now, consider deviation
    \begin{align*}
        \Delta^n \ni \vx' = \vx^* + (y^*_i - x^*_i) \ve_i + (x^*_i - y^*_i) \ve_j;
    \end{align*}
    that is, $\vx'$ is the strategy that results from $\vx$ by reallocating $(x^*_i - y^*_i)$ probability mass from action $a_i$ to action $a_j$. The difference $u(\vx', \vy^*, \vz^*) - u(\vx^*, \vy^*, \vz^*)$ can be expressed as
    \begin{align}
        &\langle \vx' - \vx^*, \mat{A} \vy^* \rangle + \frac{|\Amin|}{\epsilon} \left(z^*_i (x_i' - x_i^*) + z^*_j (x'_j - x^*_j) + z^*_{n + i}(x^*_i - x_i') + z^*_{n + j}(x^*_j - x'_j)\right) \notag \\
        &= \langle \vx' - \vx^*, \mat{A} \vy^* \rangle + \frac{|\Amin|}{\epsilon} \left(z^*_i (y_i^* - x_i^*) + z^*_j (x_i^* - y_i^*) + z^*_{n + i}(x^*_i - y_i^*) + z^*_{n + j}(y_i^* - x_i^* )\right) \notag \\
        &\leq 4 \epsilon |\Amin| + \frac{|\Amin|}{\epsilon} \left(z^*_i \cdot \left( - \frac{\epsilon}{2} \right) + \epsilon^2 \cdot 2 \epsilon + \epsilon^2 \cdot 2 \epsilon\right) \label{eq:cancel1}\\
        &\leq - \left( \frac{1}{2}z^*_i - 4 \epsilon^2-4 \epsilon \right) |\Amin|, \label{eq:cancel2}
    \end{align}
    where~\eqref{eq:cancel1} uses the following:
    \begin{itemize}
        \item $ y_i^* - x_i^* \leq - \nicefrac{\epsilon}{2}$; 
        \item $x_i^* - y_i^* \leq 2 \epsilon$ (\Cref{lemma:closeness});
        \item $z^*_{n + i} \leq \epsilon^2$ and $z_j^* \leq \epsilon^2$; and 
        \item $\langle \vx' - \vx^*, \mat{A} \vy^* \rangle \leq \| \vx' - \vx^* \|_1 \|\mat{A} \vy^* \|_\infty \leq 2 | x_i^* - y_i^* | |\Amin| \leq 4 \epsilon |\Amin| $ (since $\mat{A}$ has negative entries); 
    \end{itemize}
    Moreover, given that $(\vx^*, \vy^*, \vz^*)$ is assumed to be an $\epsilon^2$-Nash equilibrium, we have 
    \begin{align}
        \label{eq:Nashcond}
        - u(\vx', \vy^*, \vz^*) + u(\vx^*, \vy^*, \vz^*) \leq \epsilon^2.
    \end{align}
    (The utility of Player $\vx$ is given by $-u$.) Combining~\eqref{eq:Nashcond} and~\eqref{eq:cancel2},
    \begin{align}
         &\left(\frac{1}{2} z^*_i - 4 \epsilon^2-4 \epsilon\right) |\Amin| \leq  \epsilon^2 \\
        \Rightarrow \quad  z^*_i & \leq 8 \epsilon + 10\epsilon^2 \leq  9 \epsilon \textrm{ for }\epsilon \leq \frac{1}{10}.
    \end{align}
    In summary, when $x_i^* - y_i^* \geq 0$, we have shown that $z_i^* \leq 9 \epsilon$ and $z_{n+i}^* \leq \epsilon^2$. The case where $y_i^* - x_i^* \geq 0$ can be treated similarly.
\end{proof}

We next combine~\Cref{lemma:closeness,lemma:small_z} to complete the reduction from symmetric two-player games with common payoffs to $3$-player adversarial team games.

\begin{restatable}{theorem}{teamadvtosymme}
    \label{thm:teamadvtosymme}
    Given any $\epsilon^2$-Nash equilibrium $(\vx^*, \vy^*, \vz^*)$ in the adversarial team game~\eqref{eq:util-atg}, with $\epsilon \leq \nicefrac{1}{10}$, $(\vy^* , \vy^*)$ is a symmetric $(21 n + 1) |\Amin| \epsilon$-Nash equilibrium of the symmetric, two-player game $(\mat{A}, \mat{A})$ (that is, $\mat{A} = \mat{A}^\top$).
\end{restatable}

The proof is straightforward and is included in~\Cref{sec:proofs1}. \Cref{theorem:cls-completeness} restates the main complexity implication of~\Cref{theorem:hard-2cls}.

\maincls*

\begin{proof}
    $\CLS$-hardness follows directly from~\Cref{thm:teamadvtosymme} and~\Cref{theorem:hard-2cls} (due to~\citet{ghosh2024complexitysymmetricbimatrixgames}). The inclusion was shown by~\citet{Anagnostides23:Algorithms}.
\end{proof}

\subsection{A 3-player adversarial team game with a unique irrational Nash}

We begin this section by describing an interesting property for adversarial team games. Similarly as general $3$-player games \citep{Nash50:Non}, adversarial team games also admit unique irrational Nash equilibrium.

\begin{proposition}[\cite{irrationalNash}]
    \label{theorem:irrational}
    There exists an $3$-player adversarial team game with a unique Nash equilibrium that is supported on irrationals.
\end{proposition}

While previous work \citep{irrationalNash} gives the correct instance for \cref{theorem:irrational}, no proof is provided. Here we provide a relatively simple and general way to analyze the irrational Nash equilibrium in $3$-player adversarial team games as a complement of their work.

We consider a $3$-player adversarial team game in which the utility function of the adversary $u : \{1, 2\} \times \{1, 2\} \times \{1, 2\} : (\vx, \vy, \vz) \mapsto \R$ reads

\begin{center}
\renewcommand{\arraystretch}{2}
\begin{tabular}{ c|cccc } 
 \diagbox[height = 1.2cm]{$\vz$}{$(\vx, \vy)$}& (1, 1) & (1, 2) & (2, 1) & (2, 2) \\ [0.5ex]
 \toprule
 1 & $1$ & $3$ & $\frac{99}{100}$ & $- \frac{1}{100}$ \\ [0.5ex]
 2 & $\frac{9}{10}$ & $-\frac{1}{10}$ & 1 & $3$ \\ [0.5ex]
 \bottomrule
\end{tabular}
\end{center}

The proof of this result makes use of a characterization of Nash equilibria in $2 \times 2$ two-player zero-sum games, stated below; for the proof, we refer to, for example, \citet[Theorem 1.2]{sun2022propertiesnashequilibrium2}.

\begin{lemma}
    \label{lemma:2x2}
    Let $\mat{A} \in \R^{2 \times 2}$ such that
    \begin{equation*}
        ( \mat{A}_{1, 1} - \mat{A}_{1, 2} )( \mat{A}_{2, 2} - \mat{A}_{2, 1}) > 0 \text{ and } ( \mat{A}_{1, 1} - \mat{A}_{2, 1} )( \mat{A}_{2, 2} - \mat{A}_{1, 2}) > 0.
    \end{equation*}
    Then, the two-player zero-sum game $\min_{\vx \in \Delta^2} \max_{\vz \in \Delta^2} \langle \vx, \mat{A} \vz \rangle$ admits a unique (exact) Nash equilibrium with value
    \begin{equation*}
        v \defeq \frac{ \mat{A}_{1, 1} \mat{A}_{2, 2} - \mat{A}_{1, 2} \mat{A}_{2, 1}}{ \mat{A}_{1, 1} - \mat{A}_{1, 2} - \mat{A}_{2, 1} + \mat{A}_{2, 2}}.
    \end{equation*}
    Furthermore, the unique Nash equilibrium $(\vx^*, \vz^*)$ satisfies
    \begin{align*}
        & \vx^* = \left(\frac{\mat{A}_{2, 2} - \mat{A}_{2, 1}}{\mat{A}_{1, 1} - \mat{A}_{1, 2} -\mat{A}_{2, 1} + \mat{A}_{2, 2}}, \frac{\mat{A}_{1, 1} - \mat{A}_{1, 2}}{\mat{A}_{1, 1} - \mat{A}_{1, 2} -\mat{A}_{2, 1} + \mat{A}_{2, 2}}\right) \\
        & \vz^* = \left(\frac{\mat{A}_{2, 2} - \mat{A}_{1,2}}{\mat{A}_{1, 1} - \mat{A}_{1, 2} -\mat{A}_{2, 1} + \mat{A}_{2, 2}}, \frac{\mat{A}_{1, 1} - \mat{A}_{2, 1}}{\mat{A}_{1, 1} - \mat{A}_{1, 2} -\mat{A}_{2, 1} + \mat{A}_{2, 2}}\right).
    \end{align*}
\end{lemma}

\begin{proof}[Proof of~\Cref{theorem:irrational}]
    By construction of the adversarial team game, the mixed extension of the utility can be expressed as
    \begin{equation*}
        x_1 z_1 \left( y_1 + 3 y_2 \right) + x_1 z_2 \left( \frac{9}{10} y_1 - \frac{1}{10} y_2 \right) + x_2 z_1 \left( \frac{99}{100} y_1 - \frac{1}{100} y_2 \right) + x_2 z_2 \left( 3y_2 + y_1 \right).
    \end{equation*}
    Suppose that we fix $\vy \in \Delta^2$. Then, Players $\vx$ and $\vy$ are engaged in a (two-player) zero-sum game with payoff matrix
    \begin{equation}
        \label{eq:game-y}
    \mat{A}(\vy) \defeq 
        \begin{bmatrix}
        1 + 2y_2 & \frac{9}{10} - y_2 \\
        \frac{99}{100} - y_2 & 1 + 2 y_2
    \end{bmatrix}.
    \end{equation}
    We now invoke~\Cref{lemma:2x2}. Indeed, we have
    \begin{equation}
        \label{eq:diag-pos1}
        (\mat{A}(\vy)_{1,1} - \mat{A}(\vy)_{1, 2}) (\mat{A}(\vy)_{2, 2} - \mat{A}(\vy)_{2, 1}) = \left( 3 y_2 + \frac{1}{10} \right) \left( 3 y_2 + \frac{1}{100} \right) > 0
    \end{equation}
    and
    \begin{equation}
        \label{eq:diag-pos2}
        (\mat{A}(\vy)_{1,1} - \mat{A}(\vy)_{2, 1}) (\mat{A}(\vy)_{2, 2} - \mat{A}(\vy)_{1, 2}) = \left( 3 y_2 + \frac{1}{100} \right) \left( 3 y_2 + \frac{1}{10} \right) > 0;
    \end{equation}
    that is, the precondition of~\Cref{lemma:2x2} is satisfied, and so the value of~\eqref{eq:game-y} reads
    \begin{equation}
        \label{eq:value}
        v(\vy) = \min_{\vx \in \Delta^2} \max_{\vz \in \Delta^2} \langle \vx, \mat{A}(\vy) \vz \rangle = \frac{109 + 5890 y_2 + 3000 y_2^2}{110 + 6000y_2}.
    \end{equation}
    It is easy to verify that $v$ is a strictly convex function in $[0, 1]$, and admits a unique minimum corresponding to
    $
        \vy^* = 
        \left(\frac{611 - 9 \sqrt{3}}{600},
            \frac{9\sqrt{3} - 11}{600}\right),
    $
    which is irrational. Now, suppose that $(\vx^*, \vy^*, \vz^*)$ is a Nash equilibrium of the adversarial team game. We will first argue that $(\vx^*, \vz^*)$ is the unique Nash equilibrium of $\mat{A}(\vy^*)$. Indeed, suppose that there exists $\vx' \in \Delta^2$ such that $\langle \vx', \mat{A}(\vy^*) \vz^* \rangle < \langle \vx^*, \mat{A}(\vy^*) \vz^* \rangle$, or equivalently, $u(\vx', \vy^*, \vz^*) < u(\vx^*, \vy^*, \vz^*)$; this is a contradiction since $(\vx^*, \vy^*, \vz^*)$ is assumed to be a Nash equilibrium. Similar reasoning applies with respect to Player $\vz$. Thus, $(\vx^*, \vz^*)$ is a Nash equilibrium of $\mat{A}(\vy^*)$, and thereby uniquely determined by $\vy^*$---by~\Cref{lemma:2x2} coupled with \eqref{eq:diag-pos1} and~\eqref{eq:diag-pos2}. Furthermore, given the value of $\vy^*$, we get that $\vx^* = \left(\frac{3 - \sqrt{3}}{6}, \frac{3 + \sqrt{3}}{6}\right)$ and $\vz^* = \left(\frac{3 + \sqrt{3}}{6}, \frac{3 -
    \sqrt{3}}{6}\right)$. Now, consider the utility of Player $\vy$ when playing the first action $a_1$ or the second action $a_2$; plugging in the value of $\vx^*$ and $\vz^*$, we have $u(\vx^*, \ve_1, \vz^*) = \frac{578 + 9 \sqrt{3}}{600}$ and $u(\vx^*, \ve_2, \vz^*) = \frac{578 + 9 \sqrt{3}}{600}.$ Since $u(\vx^*, \ve_1, \vz^*) = u(\vx^*, \ve_2, \vz^*)$, $(\vx^*, \vy^*, \vz^*)$ is a Nash equilibrium.
    
    Moreover, suppose there exists another NE $(\vx', \vy' ,\vz')$ that is different from $(\vx^*, \vy^*, \vz^*)$. As shown above, $(\vx', \vz')$ is the unique NE of the zero-sum game induced by $\vy'$. Thus, if we have two different Nash equilibria, it implies that $\vy' \neq \vy^*$. We consider the following three cases:
    \begin{itemize}
        \item First, let $\vy'$ be a (fully) mixed strategy. Since $\vx'$ and $\vz'$ forms the unique NE in of $\mat{A}(\vy')$, we have
        \begin{align*}
            &\vx' = \left(\frac{1 + 300  y'_2}{11 + 600  y'_2}, \frac{10 + 300 y'_2}{11 + 600  y'_2}\right),\\
            &\vz' = \left(\frac{10 + 300  y'_2}{11 + 600  y'_2}, \frac{1+ 300 y'_2}{11 + 600  y'_2}\right).
        \end{align*}
        Further, for Player $\vy$,
        \begin{align*}
            &u(\vx', \ve_1, \vz') = \frac{1199 + 130800 y_2 + 3501000 y_2^2}{10(11 + 600 y_2)^2}, \\
            &u(\vx', \ve_2, \vz') = \frac{589 + 196800 y_2 + 5301000 y_2^2}{10(11 + 600 y_2)^2}.
        \end{align*}
        Since $\vy'$ is a mixed strategy, we have $u(\vx', \ve_1, \vz') = u(\vx', \ve_2, \vz')$; solving the equality we get $
        \vy' = 
        \left(\frac{611 - 9 \sqrt{3}}{600},
            \frac{9\sqrt{3} - 11}{600}\right),
    $ which contradicts the assumption that $\vy' \neq \vy^*$.
    \item If $\vy' = (1, 0)$, we have $u(\vx', \ve_1, \vz') = \frac{1199}{1210}$ and $u(\vx', \ve_2, \vz') = \frac{589}{1210}$. Thus, it follows that by unilaterally deviating to play $(0, 1)$, Player $\vy$ can decrease the utility of the adversary, contradicting the fact that $(\vx', \vy', \vz')$ is a Nash equilibrium.
    \item Finally, suppose that $\vy' = (0, 1)$. Similarly to the second case, we get $u(\vx', \ve_1, \vz') < u(\vx', \ve_2, \vz')$, which is a contradiction.
    \end{itemize}
    Thus, we conclude that $(\vx^*, \vy^*, \vz^*)$ is the unique Nash equilibrium of the 3-player adversarial team game defined above, completing the proof.   
\end{proof}

\subsection{The complexity of determining uniqueness}

A natural question arising from~\Cref{theorem:irrational} concerns the complexity of determining whether an adversarial team game admits a unique Nash equilibrium. Our next theorem establishes \NP-hardness for a version of that problem that accounts for approximate Nash equilibria.

\begin{theorem}
    \label{theorem:uniqueATG}
    For polymatrix, $3$-player adversarial team games, constants $c_1, c_2 > 0$, and $\epsilon = n^{-c_1}$, it is \NP-hard to distinguish between the following two cases:
    \begin{itemize}[noitemsep,topsep=0pt]
        \item any two $\epsilon$-Nash equilibria have $\ell_1$-distance at most $n^{-c_2}$, and
        \item there are two $\epsilon$-Nash equilibria that have $\ell_1$-distance $\Omega(1)$.
    \end{itemize}
\end{theorem}

We will see the proof of this theorem later in~\Cref{sec:nonsymmetric} when we examine the complexity of computing non-symmetric equilibria in symmetric min-max optimization problems.

%% file: text/ppadminmax.tex
\section{Complexity of equilibria in symmetric min-max optimization}

This section characterizes the complexity of computing symmetric first-order Nash equilibria\\ (\Cref{def:FONE}) in symmetric min-max optimization problems in the sense of~\Cref{def:symmetric}; namely, when $f(\vx, \vy) = - f(\vy, \vx)$ for all $(\vx, \vy) \in \calX \times \calY$ and $\calX = \calY$.

\subsection{Problem definitions and hardness results for symmetric equilibria}
\label{sec:symmetric}

Given a continuously differentiable function $f : \mathcal{D} \to \R$, we set $F_{\textrm{GDA}}:\mathcal{D} \to \mathcal{D}$ to be
$$F_{\textrm{GDA}}(\vx,\vy) \defeq \prod_{\mathcal{D}} \left[\vx - \nabla_{\vx}f(\vx,\vy),\vy + \nabla_{\vy}f(\vx,\vy)\right] \textrm{ for }(\vx,\vy)\in \mathcal{D},$$
the norm of which measures the fixed-point gap and corresponds to the update rule of GDA with stepsize equal to one; we recall that Player $\vx$ is the minimizer, while Player $\vy$ is the maximizer. The domain $\mathcal{D}$ is a compact subset of $\R^d$ for some $d \in \mathbb{N}$. Moreover, the projection operator $\prod$ is applied jointly on $\mathcal{D}$.\footnote{This is referred to as the ``safe'' version of GDA because it ensures that the mapping always lies in $\mathcal{D}$. One could also project independently on $\mathcal{D}(\vy)=\{\vx': (\vx',\vy)\in\mathcal{D}\}$ and $\mathcal{D}(\vx)=\{\vy': (\vx,\vy')\in\mathcal{D}\}$; see \citet{DSZ21} for further details and the polynomial equivalence for finding fixed points for both versions.} When $\mathcal{D}$ can be expressed as a \emph{Cartesian} product $\calX\times\calY$, the domain set is called \emph{uncoupled} (and the projection can happen independently), otherwise it is called \emph{coupled/joint}.

We begin by introducing the problem of computing fixed points of gradient descent/ascent (GDA) for domains expressed as the Cartesian product of polytopes, modifying the computational problem $\gdaFixed$ introduced by~\citet{DSZ21}.

\begin{nproblem}[\gdaFixed]
  \textsc{Input:} 
  \begin{itemize}
  \item Precision parameter $\epsilon > 0$ and smoothness parameter $L$,
  \item Polynomial-time Turing
  machine $\calC_f$ evaluating a $L$-smooth function $f : \mathcal{X} \times \mathcal{Y} \to \R$ and its gradient
  $\nabla f: \calX\times \calY \to \R^{d}$, where
   $\mathcal{X} = \{\vx:\matA_x \vx \leq \vecb_x\}$ and $\mathcal{Y} = \{\vy:\matA_y \vy \leq \vecb_y\}$ are nonempty, bounded polytopes described by input matrices
  $\matA_x \in \R^{m_x \times d_x}$, $\matA_y \in \R^{m_y \times d_y}$ and vectors $\vecb_x \in \R^{m_x}, \vecb_y \in \R^{m_y}$, with $d \defeq d_x + d_y$.
\end{itemize}
  \noindent \textsc{Output:} A point
  $(\vxstar,\vystar)\in \calX\times\calY$ such that
  $\norm{(\vx^*, \vy^*) - F_{GDA}(\vx^*,\vy^*)}_2 \leq \epsilon$.
\end{nproblem}

Based on $\gdaFixed$, we introduce the problem $\symgdaFixed$, which captures the problem of computing \emph{symmetric} (approximate) fixed points of GDA for symmetric min-max optimization problems. We would like to note that we define our computational problems as promise problems.
  
\begin{nproblem}[\symgdaFixed]
 \textsc{Input:}    \begin{itemize}
  \item Precision parameter $\epsilon > 0$ and smoothness parameter $L$, 
  \item Polynomial-time Turing
  machine $\calC_f$ evaluating a $L$-smooth, antisymmetric function $f : \mathcal{X} \times \mathcal{X} \to \R$ and its gradient
  $\nabla f: \calX\times \calX \to \R^{2d}$, where
  $\mathcal{X} = \{\vx:\matA \vx \leq \vecb\}$ is a nonempty, bounded polytope described by an input matrix
  $\matA \in \R^{m \times d}$ and vector $\vecb \in \R^{m}$.
\end{itemize}

  \noindent \textsc{Output:} A point
  $(\vxstar,\vxstar)\in \calX\times\calX$ such that
  $\norm{(\vx^*, \vx^*) - F_{GDA}(\vx^*,\vx^*)}_2 \leq \epsilon$.
\end{nproblem}

We start by showing that the problem $\symgdaFixed$ also lies in $\PPAD$; the fact that $\gdaFixed$ is in $\PPAD$---even under coupled domains---was shown to be the case by~\citet{DSZ21}.

\begin{lemma}\label{lem:membership}
  $\symgdaFixed$ is a total
search problem and lies in \PPAD.
  \end{lemma}
\begin{proof}  
We first define the function (as in~\Cref{lem:exists}) $M: \calX \to \calX$ as
\begin{equation*}
M(\vx') \defeq \prod_{\calX} \left[\vx' - \nabla _{\vx}  f(\vx,\vy)\Big|_{(\vx,\vy)=(\vx',\vx')} \right],
\end{equation*}
where we recall that $\Pi$ is the projection operator on $\calX.$ Assuming that the input function $f$ is $L$-smooth, it follows that $M(\vx')$ is $(L+1)$-Lipschitz. Furthermore, projecting on the polytope $\calX$ takes polynomial time, and so $M$ is polynomial-time computable. As a result, we can use~\citet[Proposition 2, part 2]{Etessami10:On} (see also~\citet[Proposition D.1]{Fearnley23:Complexity}), where it was shown that finding an $\epsilon$-approximate fixed point of a Brouwer function that is efficiently computable and continuous, when the domain is a bounded polytope, lies in \PPAD.
\end{proof}

Having established that $\symgdaFixed$ belongs in \PPAD, we now prove the first main hardness result of this section.

\begin{theorem}[Complexity for symmetric equilibrium]\label{thm:symmetricminmax} 
$\symgdaFixed$ is \PPAD-complete, even for quadratic functions.
\end{theorem}
\begin{proof}
We $P$-time reduce the problem of finding approximate symmetric NE in two-player symmetric games to $\symgdaFixed$.  
Given any two-player symmetric game with payoff matrices $(\mat{R},\mat{R}^{\top})$ of size $n\times n$, we set 
\begin{equation}\label{eq:matrices}
\mat{A} \defeq \frac{1}{2} \left(\mat{R}+\mat{R}^{\top}\right) \textrm{ (symmetric matrix) and } 
\mat{C} \defeq \frac{1}{2}\left(\mat{R}-\mat{R}^{\top}\right) \textrm{(skew-symmetric matrix)}.
\end{equation}

\noindent We define the \emph{quadratic}, antisymmetric function
\begin{equation}
    \label{eq:hard-quad}
 f(\vx,\vy) \defeq \frac{1}{2} \langle \vy, \mat{A}\vy \rangle - \frac{1}{2} \langle \vx, \mat{A}\vx \rangle + \langle \vy, \mat{C}\vx \rangle   
\end{equation}
with domain $\Delta^n \times \Delta^n$. Indeed, to see that $f$ is antisymmetric, one can observe that $$f(\vy,\vx) = \frac{1}{2} \langle \vx, \mat{A} \vx \rangle - \frac{1}{2} \langle \vy, \mat{A} \vy \rangle + \langle \vx, \mat{C} \vy \rangle = \frac{1}{2} \langle \vx, \mat{A}\vx \rangle -\frac{1}{2} \langle \vy, \mat{A} \vy \rangle - \langle \vy, \mat{C}^{\top} \vx \rangle = -f(\vx,\vy).$$
Assuming that all entries of $\mat{R}$ lie in $[-1,1],$ it follows that the singular values of $\mat{A}$ and $\mat{C}$ are bounded by $n.$ As a result $f$ and $\nabla_{\vx} f = -\mat{A}\vx-\mat{C}\vy, \nabla_{\vy} f = \mat{A}\vy +\mat{C}\vx$ are polynomial time computable and continuous, and $\nabla_{\vx} f, \nabla_{\vy}f$ are $L$-Lipschitz for $L \leq 2n,$ thus $f$ is $4n$-smooth.

We assume $\vx$ is the minimizer and $\vy$ is the maximizer, and let $(\vx^*,\vx^*)$ be an $\epsilon$-approximate fixed point of GDA. We shall show that $(\vx^*,\vx^*)$ is an $4n\epsilon$-approximate NE of the symmetric two-player game $(\mat{R},\mat{R}^{\top})$. Since $(\vx^*,\vx^*)$ is an $\epsilon$-approximate fixed point of GDA, we can use~\Cref{lem:approxsmooth} (\Cref{sec:proofs2}) and obtain the following variational inequalities:
\begin{equation*}
\max_{\vx^*+\bm{\delta} \in \Delta^n, \norm{\bm{\delta}}_2\leq 1} \bm{\delta}^{\top} (\mat{A}\vx^*+\mat{C}\vx^*)\leq \epsilon\left(2n+1\right),
\end{equation*}
implying that (since the diameter of $\Delta^n$ is $\sqrt{2}$ in $\ell_2$)
\begin{equation}
\tag{VI for NE}\label{eq:VIforNE}
\langle \vx - \vx^*,(\mat{A}+\mat{C})\vx^*  \rangle \leq \sqrt{2}\epsilon\left(2n+1\right) \textrm{ for any }\vx\in\Delta^n. 
\end{equation}

Now, we observe that \eqref{eq:VIforNE} implies that $(\vx^*,\vx^*)$ is a $\sqrt{2}\epsilon\left(2n+1\right)$-approximate symmetric NE in the two-player symmetric game with payoff matrices $(\mat{A}+\mat{C},\mat{A}-\mat{C})$ (recall Definition \eqref{def:NE}). Since $\sqrt{2}\epsilon\left(2n+1\right) \leq 4n\epsilon$ for $n\geq 2$, our claim follows.

By~\Cref{theorem:PPAD_for _symmetric} and~\Cref{lem:membership}, we conclude that $\symgdaFixed$ is \PPAD-complete, even for quadratic functions that are $O(n)$-smooth, $O(n)$-Lipschitz and $\epsilon \leq \nicefrac{1}{n^{1+c}}$, for any $c>0$.
\end{proof}

For symmetric first-order Nash equilibria, the same argument establishes \PPAD-hardness for any $\epsilon \leq \nicefrac{1}{n^c}$, where $c > 0$ (as claimed in~\Cref{theorem:main}). Moreover, leveraging the hardness result of~\citet{Rubinstein16:Settling}, we can also immediately obtain constant inapproximability under the so-called \emph{exponential-time hypothesis (ETH)} for \PPAD---which postulates than any algorithm for solving \textsc{EndOfALine}, the prototypical \PPAD-complete problem, requires $2^{\tilde{\Omega}(n)}$ time.

\begin{corollary}
    \label{cor:constant}
    Computing an $\Theta(1)$-approximate first-order Nash equilibrium in symmetric $n$-dimensional min-max optimization requires $n^{\tilde{\Omega}(\log n)}$ time, assuming ETH for \PPAD.
\end{corollary}

\begin{remark}[Comparison with \citet{DSZ21}] The argument of~\Cref{thm:symmetricminmax} can be slightly modified to imply one of the main results of~\citet{DSZ21}---with simplex constraints instead of box constraints. We provide a simple proof of this fact below (\Cref{thm:simple}). The main idea is to introduce coupled constraints in order to \emph{force symmetry}, that is, constraints of the form $-\delta \leq x_i - y_i \leq \delta$ for all $i \in [n]$, where, if $\epsilon$ is the approximation accuracy, $\delta$ is of order $\Theta\left(\epsilon^{1/4}\right)$. Our result pertaining to symmetric equilibria is stronger in that it accounts for deviations in the whole domain, not merely on the coupled feasibility set. 
\end{remark}

\begin{theorem}[\PPAD-hardness for coupled domains]\label{thm:simple} The problem $\gdaFixed$ is \PPAD-hard when the domain is a joint polytope, even for quadratic functions.
\end{theorem}
\begin{proof}
The proof follows similar steps with \Cref{thm:symmetricminmax}, namely, we $P$-time reduce the problem of finding approximate symmetric NE in two-player symmetric games to $\gdaFixed$ with coupled domains.  
Given a two-player symmetric game with payoff matrices $(\mat{R},\mat{R}^{\top})$ of size $n\times n$, we set 
$\mat{A} \defeq \frac{1}{2} \left(\mat{R}+\mat{R}^{\top}\right)$, 
$\mat{C} \defeq \frac{1}{2}\left(\mat{R}-\mat{R}^{\top}\right)$ and define again the quadratic, antisymmetric function
\begin{equation*}
 f(\vx,\vy) \defeq \frac{1}{2} \langle \vy, \mat{A}\vy \rangle - \frac{1}{2} \langle \vx, \mat{A}\vx \rangle + \langle \vy, \mat{C}\vx \rangle.   
\end{equation*}
Moreover, given a parameter $\delta > 0$ (to be specified shortly), we define the joint  domain of $f$ to be
\begin{equation}\label{eq:polytope}\tag{joint Domain}
\mathcal{D} := \left\{(\vx,\vy) \in \Delta^n \times \Delta^n: -\delta\leq x_i - y_i \leq \delta \textrm{ for all }i\in[n]\right\}.
\end{equation}
Let $(\vx^*,\vy^*)$ be an $\epsilon$-approximate fixed point of GDA. We will show that $\left(\frac{\vx^*+\vy^*}{2},\frac{\vx^*+\vy^*}{2}\right)$ is an $O(\epsilon^{1/4})$-approximate (symmetric) NE of the game $(\mat{R},\mat{R}^{\top})$ for an appropriate choice of $\delta.$ 

\smallskip

We set $\mathcal{D}(\vx^*) = \{\vy:(\vx^*,\vy)\in\mathcal{D}\}$ and 
$\mathcal{D}(\vy^*) = \{\vx:(\vx,\vy^*)\in\mathcal{D}\}.$
In words, $\mathcal{D}(\vx^*)$ and $\mathcal{D}(\vy^*)$ capture the allowed deviations for $\vy$ and $\vx$ respectively. It also holds that $f$ is $G$-Lipschitz continuous with $G=4n$ and also $4n$-smooth (using the same reasoning as in Theorem~\ref{thm:symmetricminmax}). 

\smallskip
\noindent Since $(\vx^*,\vy^*)$ is an $\epsilon$-approximate fixed point of GDA, using~\Cref{lem:safe} (\Cref{sec:proofs2}), the following VIs must hold for some positive constant $K<10$ and $n$ sufficiently large:

\begin{equation}
\label{eq:VIforcoupled}
\begin{array}{cc}
\langle \vx - \vx^*,-\mat{A}\vx^*+\mat{C}^{\top}\vy^*  \rangle \geq -Kn^{3/2}\sqrt{\epsilon} \textrm{ for any }\vx \in \mathcal{D}(\vy^*) \textrm{ and }\\ 
\langle \vy - \vy^*,\mat{A}\vy^*+\mat{C}\vx^*  \rangle \leq Kn^{3/2}\sqrt{\epsilon} \textrm{ for any }\vy \in \mathcal{D}(\vx^*).
\end{array}
\end{equation}

\noindent Let $\overline{\mathcal{D}} = \left\{\vz \in \Delta^n: \left\|\vz-\frac{\vx^*+\vy^*}{2}\right\| _{\infty}\leq \frac{\delta}{2}\right\}.$ By triangle inequality, it follows that 
$\overline{\mathcal{D}} \subseteq \mathcal{D}(\vy^*)\cap\mathcal{D}(\vx^*).$
We express the VIs of \eqref{eq:VIforcoupled} using a single variable $\vz$ and common deviation domain:
\begin{equation*}
\label{eq:VIforcoupledz}
\begin{array}{cc}
\langle \vz - \vx^*,-\mat{A}\vx^*+\mat{C}^{\top}\vy^*  \rangle \geq -Kn^{3/2}\sqrt{\epsilon} \textrm{ and } 
\langle \vz - \vy^*,\mat{A}\vy^*+\mat{C}\vx^*  \rangle \leq Kn^{3/2}\sqrt{\epsilon} \textrm{ for any }\vz \in \overline{\mathcal{D}}.
\end{array}
\end{equation*}
Multiplying the first inequality by $-1/2$ and the second with $1/2$ and adding them up gives
\begin{equation}
\label{eq:touse}
\left\langle \vz - \frac{\vx^*+\vy^*}{2}, (\mat{A}+\mat{C})\frac{\vx^*+\vy^*}{2} \right\rangle \leq \frac{1}{4}\left\langle \vx^*-\vy^*,\mat{A}(\vx^*-\vy^*)\right\rangle + Kn^{3/2}\sqrt{\epsilon}. 
\end{equation}
Since $\vx^*,\vy^* \in \mathcal{D}$, it follows that $\left\langle \vx^*-\vy^*,\mat{A}(\vx^*-\vy^*)\right\rangle \leq n \|\vx^*-\vy^*\|^2_2 \leq n^2 \delta^2$. 
Combining this fact with \eqref{eq:touse}, we conclude that
\begin{equation}
\tag{VImedian}
\label{eq:lastVINE}
\left\langle \vz - \frac{\vx^*+\vy^*}{2}, (\mat{A}+\mat{C})\frac{\vx^*+\vy^*}{2} \right\rangle \leq n^2 \delta^2 + Kn^{3/2}\sqrt{\epsilon} \textrm{ for any }\vz \in \overline{\mathcal{D}}. 
\end{equation}

\eqref{eq:lastVINE} shows that by deviating from $\left(\frac{\vx^*+\vy^*}{2},\frac{\vx^*+\vy^*}{2}\right)$ to some $\vz$ in $\overline{\mathcal{D}}$, the payoff cannot increase by more than $\left(n^2\delta^2 + Kn^{3/2}\sqrt{\epsilon}\right)$ in the two-player symmetric game with matrices $(\mat{R},\mat{R}^{\top})$.

\noindent We consider any pure strategy $\bm{e}_j$ for $j\in[n]$. If $\left\|\bm{e}_j - \frac{\vx^*+\vy^*}{2}\right\|_{\infty} \leq \frac{\delta}{2}$ then $\bm{e}_j \in \overline{\mathcal{D}}$ 
and it is captured by \eqref{eq:lastVINE}. Suppose that 
$\left\|\bm{e}_j - \frac{\vx^*+\vy^*}{2}\right\|_{\infty} > \frac{\delta}{2}$
and consider the point $\vz' \in \overline{\mathcal{D}}$ on the line segment between $\bm{e}_j$ and $\frac{\vx^*+\vy^*}{2}$ that intersects the boundary of $\overline{\mathcal{D}}.$ It holds that $\bm{e}_j - \frac{\vx^*+\vy^*}{2} = c\left(\vz' - \frac{\vx^*+\vy^*}{2} \right)$ for some positive $c \leq \frac{2}{\delta}$ (it cannot be larger because otherwise the infinity norm of the difference between $\bm{e}_j$ and $\frac{\vx^*+\vy^*}{2}$ would exceed one, which is impossible as they both belong to $\Delta^n$). Therefore, 
\begin{equation}
\label{eq:lastlastVINE}
\left\langle \bm{e}_j - \frac{\vx^*+\vy^*}{2}, (\mat{A}+\mat{C})\frac{\vx^*+\vy^*}{2} \right\rangle \leq 2n^2 \delta + \frac{2Kn^{3/2}\sqrt{\epsilon}}{\delta} \textrm{ for any pure strategy }j. 
\end{equation}

From \eqref{eq:lastlastVINE}, we conclude that $\frac{\vx^*+\vy^*}{2}$ is $\left(2n^{2}\delta + \frac{2Kn^{3/2}\sqrt{\epsilon}}{\delta}\right)$-approximate NE of the symmetric two-player game $(\mat{R},\mat{R}^{\top}).$ We choose $\delta = \epsilon^{1/4} n^{-1/4}$ and we get that $\frac{\vx^*+\vy^*}{2}$ is an $O(n^{7/4}\epsilon^{1/4})$-approximate NE for $(\mat{R},\mat{R}^{\top})$, and thus the hardness result holds for $\epsilon$ of order  $O\left(\frac{1}{n^{7+c}}\right)$ for any constant $c>0.$ We note that if instead of an $\epsilon$-approximate fixed point of GDA, we were given an $\epsilon$-approximate First-order NE, then the hardness result would hold for any $\epsilon$ of order $\frac{1}{n^{c}}$ with $c>0.$
\end{proof}

\paragraph{Hardness results for symmetric dynamics}

Another interesting consequence of~\Cref{thm:symmetricminmax} is that it immediately precludes convergence under a broad class of iterative algorithms in general min-max optimization problems. 

\begin{definition}[Symmetric learning algorithms for min-max]
    \label{def:sym-dynamics}
    Let $T \in \mathbb{N}$. A deterministic, polynomial-time learning algorithm $\calA$ proceeds as follows for any time $t \in [T]$. It outputs a strategy as a function of the history $\mathcal{H}^{(t)}$ it has observed so far (where $\mathcal{H}^{(1)} \defeq \emptyset$ ), and then receives as feedback $\vec{g}^{(t)}$. It then updates $\mathcal{H}^{(t+1)} \defeq ( \mathcal{H}^{(t)}, \vec{g}^{(t)})$. 
    
    A \emph{symmetric} learning algorithm in min-max optimization consists of Player $\vx$ employing algorithm $\calA$ with history $\mathcal{H}_x^{(t)} \defeq (\nabla_{\vx} f(\vx^{(t)}, \vy^{(t)}) )_{t=1}^T$, and Player $\vy$ employing the \emph{same} algorithm with history $\mathcal{H}_y^{(t)} \defeq (- \nabla_{\vy} f(\vx^{(t)}, \vy^{(t)}) )_{t=1}^T$.
\end{definition}

(A consequence of the above definition is that both players initialize from the same strategy.) Many natural and well-studied algorithms in min-max optimization adhere to~\Cref{def:sym-dynamics}. Besides the obvious example of gradient descent/ascent, we mention extragradient descent(/ascent), optimistic gradient descent(/ascent), and optimistic multiplicative weights---all assumed to be executed simultaneously. A simple non-example is \emph{alternating} gradient descent(/ascent)~\citep{Wibisono22:Alternating,Bailey20:Finite}, wherein players do not update their strategies simultaneously.

\begin{theorem}
    \label{theorem:sym-dyn}
    No symmetric learning algorithm (per~\Cref{def:sym-dynamics}) can converge to $\epsilon$-first-order Nash equilibria in min-max optimization in polynomial time when $\epsilon = \nicefrac{1}{n^c}$, unless $\PPAD = \P$.
\end{theorem}

Indeed, this is a direct consequence of our argument in~\Cref{thm:symmetricminmax}: under~\Cref{def:sym-dynamics} and the min-max optimization problem~\eqref{eq:hard-quad}, it follows inductively that $\vx^{(t)} = \vy^{(t)}$ and $\mathcal{H}_x^{(t)} = \mathcal{H}_y^{(t)}$ for all $t \in [T]$. But~\Cref{thm:symmetricminmax} implies that computing a symmetric first-order Nash equilibrium is $\PPAD$-hard when $\epsilon = \nicefrac{1}{n^c}$.

Assuming that $\P \neq \PPAD$, \Cref{theorem:sym-dyn}, and in particular its instantiation in team zero-sum games (\Cref{theorem:team-hard}), recovers and significantly generalizes some impossibility results shown by~\citet{kalogiannis2021teamwork} concerning lack of convergence for certain algorithms, such as optimistic gradient descent(/ascent)---our hardness result goes much further, precluding any algorithm subject to~\Cref{def:sym-dynamics}, albeit being conditional.

\subsection{The complexity of non-symmetric fixed points}
\label{sec:nonsymmetric}

An immediate question raised by~\Cref{thm:symmetricminmax} concerns the computational complexity of finding \emph{non-symmetric} fixed points of GDA for symmetric min-max optimization problems. Since totality is not guaranteed, unlike $\symgdaFixed$, we cannot hope to prove membership in \PPAD. In fact, we show that finding a non-symmetric fixed point of GDA is \FNP-hard (\Cref{theorem:non-symmetric}). To do so, we first define formally the computational problem of interest.

\begin{nproblem}[\nsymgdaFixed]
 \textsc{Input:}  \begin{itemize}
   \item Parameters $\epsilon,\delta>0$ and Lipschitz constant $L$ and
  \item Polynomial-time Turing
  machine $\calC_f$ evaluating a $L$-smooth antisymmetric function $f : \mathcal{X} \times \mathcal{X} \to \R$ and its gradient
  $\nabla f: \calX\times \calX \to \R^{2d}$, where
  $\mathcal{X} = \{\vx:\matA \vx \leq \vecb\}$ is a nonempty, bounded polytope described by a matrix
  $\matA \in \R^{m \times d}$ and vector $\vecb \in \R^{m}$.
\end{itemize}

  \noindent \textsc{Output:} A point
  $(\vxstar,\vystar)\in \calX\times\calX$ such that $\norm{\vxstar-\vystar}_2 \geq \delta$ and
  $\norm{(\vx^*, \vy^*) - F_{GDA}(\vx^*,\vy^*)}_2 \leq \epsilon$ if it exists, otherwise return \textsf{NO}.
\end{nproblem}

\noindent We establish that $\nsymgdaFixed$ is \FNP-hard. Our main result is restated below.

\nonsymmetric*

Our reduction builds on the hardness result of~\citet{MCLENNAN2010683}---in turn based on earlier work by~\citet{gilboa1989nash,Conitzer08:New}---which we significantly refine in order to account for $\poly(1/n)$-Nash equilibria. We begin by describing their basic approach. Let $G = ([n], E)$ be an $n$-node, undirected, unweighted graph, and construct
\begin{equation}
    \label{eq:matA}
    \mat{A}_{i, j} = 
    \begin{cases}
        -1 & \text{if $i = j$}, \\
        0 & \text{if $\{i, j\} \in E$}, \\
        -2 & \text{otherwise.}
    \end{cases}
\end{equation}
(\Cref{fig:graph} depicts an illustrative example.) Based on this matrix, \citet{MCLENNAN2010683} consider the symmetric, identical-payoff, two-player game $(\mat{A}, \mat{A})$---by construction, $\mat{A} = \mat{A}^\top$, and so this game is indeed symmetric. They were able to show the following key property.
\input{figs/graph}

\begin{lemma}[\citep{MCLENNAN2010683}]
    \label{lemma:Nashgap}
    Let $C_k \subseteq [n]$ be a maximum clique of $G$ with size $k$ and $\vx^* = \frac{1}{k} \sum_{i \in C_k} \ve_{i}$. Then, $(\vx^*, \vx^*)$ is a Nash equilibrium of $(\mat{A}, \mat{A})$ that attains value $-\frac{1}{k}$. Furthermore, any symmetric Nash equilibrium not in the form described above has value at most $- \frac{1}{k - 1}$.
\end{lemma}
The idea now is to construct a new symmetric, identical-payoff game $(\mat{B}, \mat{B})$, for
    \begin{align}
        \mat{B} \defeq \begin{bmatrix}
         \mat{A}_{1, 1} & \cdots & \mat{A}_{1, n}  & r\\
         \vdots & \ddots & \vdots & \vdots \\
         \mat{A}_{n, 1} & \cdots & \mat{A}_{n,n} & r\\
         r & \cdots & r  & V\\
    \end{bmatrix}, \label{eq:unique_NP}
    \end{align}
    where $V \defeq - \frac{1}{k}$ and $r = \frac{1}{2} (- \frac{1}{k} - \frac{1}{k - 1}) = -\frac{2k-1}{2(k - 1)k}$. Coupled with~\Cref{lemma:Nashgap}, this new game yields the following \NP-hardness result.

\begin{theorem}[\citep{MCLENNAN2010683}]
    \label{theorem:knownNP}
    It is \NP-hard to determine whether a symmetric, identical-payoff, two-player game has a unique symmetric Nash equilibrium.
\end{theorem}

Our goal here is to prove a stronger result, \Cref{theorem:symmetric-new}, that characterizes the set of $\epsilon$-Nash equilibria even for $\epsilon = \nicefrac{1}{n^c}$; this will form the basis for our hardness result in min-max optimization and adversarial team games. To do so, we first derive some basic properties of game~\eqref{eq:unique_NP}.

Game~\eqref{eq:unique_NP} always admits the trivial (symmetric) Nash equilibrium $(\vec{e}_{n+1}, \vec{e}_{n+1})$. Now, consider any symmetric Nash equilibrium $(\vx^*, \vx^*)$ with $x^*_{n+1} \neq 1$. If $x^*_{n+1} = 0$, it follows that $(\vx^*_{[i \cdots n] }, \vx^*_{[i \cdots n]})$ is a Nash equilibrium of $(\mat{A}, \mat{A})$, which in turn implies that $G$ admits a clique of size $k$; this follows from~\Cref{lemma:Nashgap}, together with the fact that $- \nicefrac{1}{k - 1} < r < -\nicefrac{1}{k}$.

We now analyze the case where $x^*_{n+1} \in (0, 1)$. It then follows that $(\nicefrac{\vx^*_{[1 \cdots n]}}{1 - x^*_{n+1}}, \nicefrac{\vx^*_{[1 \cdots n]}}{1 - x^*_{n+1}})$ is a (symmetric) Nash equilibrium of $(\mat{A}, \mat{A})$. Furthermore, the utility of playing action $a_{n+1}$ is $(1 - x^*_{n+1}) r + x^*_{n+1} V > r $. By~\Cref{lemma:Nashgap}, it follows that $(\nicefrac{\vx^*_{[1 \cdots n] }}{1 - x^*_{n+1}}, \nicefrac{\vx^*_{[1 \cdots n]}}{1 - x^*_{n+1}})$ has a value of $V$ and $G$ admits a clique of size $k$. As a result, the utility of playing any action $a_i$, with $i \in \text{supp}(\vx^*)$ and $i \neq n+1$, is $( 1 - x^{*}_{n+1}) V + x^*_{n+1} r$. At the same time, the utility of playing action $a_{n+1}$ reads $(1 - x_{n+1}^*) r + x_{n+1}^* V$. Equating those two quantities, it follows that $x^*_{n+1} = \nicefrac{1}{2}$.

In summary, $G$ contains a clique of size $k$ if and only if game~\eqref{eq:unique_NP} admits a unique symmetric Nash equilibrium, which implies~\Cref{theorem:knownNP}. What is more, we have shown a stronger property. Namely, any symmetric Nash equilibrium of $(\mat{B}, \mat{B})$ has to be in one of the following forms:
\begin{enumerate} \label{eq:three_exact_NE}
        \item $(\vx^*, \vx^*)$ with $\vx^* \defeq \ve_{n+ 1}$;\label{item:case1}
        \item $(\vx^*, \vx^*)$ with $\vx^* \defeq \frac{1}{k} \sum_{i \in C_k} \ve_i$, where $C_k \subseteq [n]$ is a clique in $G$ of size $k$;\label{item:case2}
        \item $(\vx^*, \vx^*)$ with $\vx^* \defeq \frac{1}{2} \ve_{n + 1} + \frac{1}{2k} \sum_{i \in C_k} \ve_i$, where $C_k \subseteq [n]$ is a clique in $G$ of size $k$.\label{item:case3}
\end{enumerate}
In particular, the equilibria in Items~\ref{item:case2} or~\ref{item:case3}---which exist iff $G$ contains a clique of size $k$---are always far from the one in Item~\ref{item:case1}. However, this characterization only applies to exact Nash equilibria. In any two-player game $\Gamma$, when $\epsilon$ is sufficiently small with $\log(1/\epsilon) \leq \poly(|\Gamma|)$,  \citet{Etessami10:On} have shown that any $\epsilon$-Nash equilibrium is within $\ell_1$-distance $\delta$ from an exact one, and so the above characterization can be applied; unfortunately, this does not apply (for general games) in the regime we are interested, namely $\epsilon = \poly(1/n)$.

We address this challenge by refining the result of~\citet{MCLENNAN2010683}. Our main result, which forms the basis for~\Cref{theorem:non-symmetric} and~\Cref{theorem:uniqueATG}, is summarized below.

\begin{theorem}
    \label{theorem:symmetric-new}
    For symmetric, identical-interest, two-player games, constants $c_1, c_2 > 0$, and $\epsilon = n^{-c_1}$, it is \NP-hard to distinguish between the following two cases:
    \begin{itemize}[noitemsep,topsep=0pt]
        \item any two symmetric $\epsilon$-Nash equilibria have $\ell_1$-distance at most $n^{-c_2}$, and
        \item there are two symmetric $\epsilon$-Nash equilibria that have $\ell_1$-distance $\Omega(1)$.
    \end{itemize}
\end{theorem}

Our reduction proceeds similarly, but defines $\lineA$ to be the adjacency matrix of $G$ with $\delta \in (0, 1)$ in each diagonal entry. Using $\lineA$, we show that we can refine~\Cref{lemma:Nashgap} of~\citet{MCLENNAN2010683}. Before we state the key property we prove in \Cref{lemma:well_supported_nash_value}, we recall the following definition.

\begin{definition}[Well-supported NE]
    A symmetric strategy profile $(\vx, \vx)$ is an \emph{$\epsilon$-well-supported} Nash equilibrium of the symmetric, identical-payoff game $(\lineA, \lineA)$ if for all $i \in [n]$,
    \begin{align*}
        x_i > 0 \implies (\lineA \vx)_i \geq \max_{j \in [n]} (\lineA \vx)_j - \epsilon.
    \end{align*}
\end{definition}

\begin{restatable}{lemma}{wellsupclique}
\label{lemma:well_supported_nash_value}
    Suppose that the maximum clique in $G$ is of size $k$. For any symmetric $\epsilon$-well-supported NE $(\hat{\vx}, \hat{\vx})$ of $(\lineA, \lineA)$ not supported on a clique of size $k$, we have $u(\hat{\vx}, \hat{\vx}) \leq 1 - \frac{1}{k} + \frac{\delta}{k} - \frac{2\delta}{ n^2k^4} + 2 \epsilon$.
\end{restatable}

Equipped with this property, we show (in~\Cref{sec:proofs3}) that a similar argument to the one described earlier concerning game~\eqref{eq:unique_NP} establishes~\Cref{theorem:symmetric-new}. We proceed now with~\Cref{theorem:non-symmetric}.
\begin{proof}[Proof of~\Cref{theorem:non-symmetric}]
It suffices to consider the antisymmetric function $f(\vx,\vy) \defeq \vy^{\top} \mat{B}\vy - \vx^{\top} \mat{B}\vx$, where symmetric matrix $\mat{B}$ is defined as in \eqref{eq:unique_NP}, using our new matrix $\lineA$ instead of $\mat{A}$ (see \eqref{eq:unique_NP_matrix}). Any $\epsilon$-first-order Nash equilibrium $(\vx^*, \vy^*)$ of this (separable) min-max optimization problem induces, two symmetric $\epsilon$-Nash equilibria---namely, $(\vx^*, \vx^*)$ and $(\vy^*, \vy^*)$---in the symmetric, identical-interest, game $(\mat{B}, \mat{B})$. Using~\Cref{theorem:symmetric-new}, the claim follows.
\end{proof}

Finally, the proof of~\Cref{theorem:uniqueATG} that was claimed earlier follows immediately by combining~\Cref{theorem:symmetric-new} with the reduction of~\Cref{sec:cls}, and in particular, \Cref{lemma:closeness,lemma:small_z}.

\subsection{Team zero-sum games}
\label{sec:teamzero}

Our previous hardness result concerning symmetric min-max optimization problems does not have any immediate implications for (normal-form) team zero-sum games since the class of hard instances we constructed earlier contains a quadratic term. Our next result provides such a hardness result by combining the basic gadget we introduced in~\Cref{sec:cls} in the context of adversarial team games; the basic pieces of the argument are similar to the ones we described in~\Cref{sec:cls}, and so the proof is deferred to~\Cref{sec:proofs4}. Our goal is to prove the following.

\teamhard*

We begin by describing the class of $3$ vs. $3$ team zero-sum games upon which our hardness result is based on. To do so, based on~\eqref{eq:util-atg}, let us define the auxiliary function
\begin{equation*}
    \delta: \Delta^n \times \Delta^n \times \Delta^{2n} \ni (\vx, \vy, \vz) \mapsto \frac{|\Amin|}{\epsilon} \sum_{i=1}^n ( z_i (x_i - y_i) + z_{n+i} (y_i - x_i)) + |\Amin| z_{2n+1}.
\end{equation*}
In what follows, the $3$ players of the one team will be identified with $(\vx, \vy, \vz)$, while the $3$ players of the other team with $(\hvx, \hvy, \hvz)$. Now, we define the utility of the latter team to be
\begin{equation}
    \label{eq:team-sym}
        u(\vx, \vy, \vz , \hat{\vx}, \hat{\vy}, \hat{\vz})  = \langle \vx, \mat{A} \vy \rangle - \langle \hvx, \mat{A} \hvy \rangle + \langle \vx, \mat{C} \hvx \rangle  + \delta(\vx, \vy, \hvz) - \delta(\hvx, \hvy, \vz),
\end{equation}
where again $\mat{A}$ is a symmetric matrix and $\mat{C}$ is skew-symmetric. As in~\Cref{sec:cls}, it is assumed for convenience that $\mat{A}_{i, j} \leq -1$ for all $i, j \in [n]$; we denote by $\Amin$ the minimum entry of $\mat{A}$. That the game defined above is symmetric is clear: $u(\vx, \vy, \vz, \hvx, \hvy, \hvz) = - u(\hvx, \hvy, \hvz, \vx, \vy, \vz)$ for all joint strategies (since $\mat{C} = - \mat{C}^\top$). It is also evident that~\eqref{eq:team-sym} is a polymatrix game, as promised.

The first key lemma, which mirrors~\Cref{lemma:closeness}, shows that, in an approximate Nash equilibrium of~\eqref{eq:team-sym}, $\vx \approx \vy$ and $\hvx \approx \hvy$. This is crucial as it allows us to construct---up to some small error---quadratic terms in the utility function, as in our hardness result for symmetric min-max optimization.

\begin{restatable}{lemma}{closeteams}
    \label{lemma:close-teams}
    Let $(\vx^*, \vy^*, \vz^*, \hvx^*, \hvy^*, \hvz^*)$ be an $\epsilon^2$-Nash equilibrium of~\eqref{eq:team-sym} with $\epsilon^2 \leq \nicefrac{1}{2}$. Then, $\| \vx^* - \vy^* \|_\infty \leq 2 \epsilon$ and $\| \hvx^* - \hvy^* \|_\infty \leq 2 \epsilon$.
\end{restatable}

Next, following the argument of~\Cref{lemma:small_z}, we show that, in equilibrium, Players $\vz$ and $\hvz$ place most of their probability mass on action $a_{2n + 1}$, thereby having only a small effect on the game between Players $\vx$ and $\vy$ vs. $\hvx$ and $\hvy$.

\begin{restatable}{lemma}{smallz}
    \label{lemma:smallzteam}
    Let $(\vx^*, \vy^*, \vz^*, \hvx^*, \hvy^*, \hvz^*)$ be an $\epsilon^2$-Nash equilibrium of~\eqref{eq:team-sym} with $\epsilon \leq \nicefrac{1}{10}$. Then, $z_j, \hat{z}_j \leq 9 \epsilon$ for all $j \in [2n]$.
\end{restatable}

Armed with those two basic lemmas, we complete the proof of~\Cref{theorem:team-hard} in~\Cref{sec:proofs4}.

%% file: figs/graph.tex
\begin{figure}
    \centering
    \begin{tikzpicture}

        \begin{scope}[xshift=-3cm, yshift=.6cm, scale=0.8]
            \foreach \i in {1, 2, 3, 4, 5} {
                \node[circle, draw, minimum size=0.3cm, font=\footnotesize] (v\i) at (360/5 * \i:1.5) {\i};
            }

            \draw (v1) -- (v2);
            \draw (v1) -- (v3);
            \draw (v1) -- (v4);
            \draw (v2) -- (v3);
            \draw (v2) -- (v4);
            \draw (v3) -- (v4);
            \draw (v3) -- (v5);
            \draw (v4) -- (v5);
        \end{scope}

        \begin{scope}[xshift=1.5cm, yshift=-0.75cm]
            \foreach \i in {1, 2, 3, 4, 5} {
                \node[circle, draw, minimum size=0.3cm, font=\footnotesize,inner sep=0pt] at (-0.5, 2.5 - 0.5*\i + 0.25) {\i};
                \node[circle, draw, minimum size=0.3cm, font=\footnotesize,inner sep=0pt] at (0.5*\i - 0.25, 3.0) {\i};
            }

            \node at (0.25, 2.25) {$-1$}; 
            \node at (0.75, 2.25) {$0$};  
            \node at (1.25, 2.25) {$0$};  
            \node at (1.75, 2.25) {$0$};  
            \node at (2.25, 2.25) {$-2$}; 

            \node at (0.25, 1.75) {$0$};  
            \node at (0.75, 1.75) {$-1$}; 
            \node at (1.25, 1.75) {$0$};  
            \node at (1.75, 1.75) {$0$};  
            \node at (2.25, 1.75) {$-2$}; 

            \node at (0.25, 1.25) {$0$};  
            \node at (0.75, 1.25) {$0$};  
            \node at (1.25, 1.25) {$-1$}; 
            \node at (1.75, 1.25) {$0$};  
            \node at (2.25, 1.25) {$0$};  

            \node at (0.25, 0.75) {$0$};  
            \node at (0.75, 0.75) {$0$};  
            \node at (1.25, 0.75) {$0$};  
            \node at (1.75, 0.75) {$-1$}; 
            \node at (2.25, 0.75) {$0$};  

            \node at (0.25, 0.25) {$-2$}; 
            \node at (0.75, 0.25) {$-2$}; 
            \node at (1.25, 0.25) {$0$};  
            \node at (1.75, 0.25) {$0$};  
            \node at (2.25, 0.25) {$-1$}; 
        \end{scope}

    \end{tikzpicture}
    \caption{An example of matrix $\mat{A} = \mat{A}(G)$ (right) for graph $G$ (left).}
    \label{fig:graph}
\end{figure}
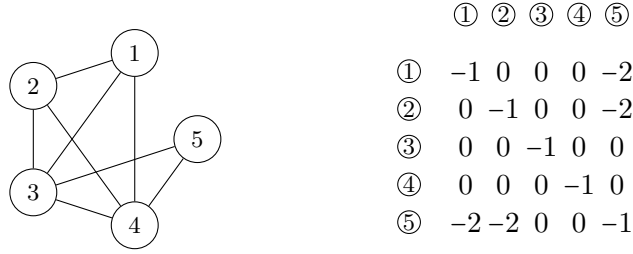

%% file: text/conclusions.tex
\section{Conclusion and open problems}

We have provided a number of new complexity results concerning min-max optimization in general, and team zero-sum games in particular (\emph{cf.}~\Cref{tab:results}). There are many interesting avenues for future research. The complexity of computing first-order Nash equilibria (equivalently, the $\gdaFixed$ problem) remains wide open, but our hardness results suggest a possible approach: as we have seen, in symmetric min-max optimization, computing \emph{either} symmetric or non-symmetric equilibria is intractable. It would be enough if one could establish this using the same underlying function---that is, somehow combine our two reductions into one. With regard to Nash equilibria in adversarial team games, the main question that remains open concerns the \emph{strong} approximation version of the problem, whereby one requires identifying points that are close---in geometric distance---to exact Nash equilibria. Is there some natural subclass of \textsf{FIXP} that characterizes the complexity of that problem? Relatedly, \citet{Etessami10:On} provided strong evidence for the intractability of strong approximation---hovering above \NP---through a reduction from the \emph{square root sum} problem; it is currently unclear whether such reductions can carry over using adversarial team games.

%% file: text/appendix.tex
\section{Omitted proofs}
\label{sec:proofs}

This section contains the proofs omitted from the main body.

\subsection{Proofs from Section~\ref{sec:cls}}
\label{sec:proofs1}

We begin with~\Cref{lemma:mass_on_small_reward}.

\begin{proof}[Proof of~\Cref{lemma:mass_on_small_reward}]
    For the sake of contradiction, suppose that $x^*_i(a_k) > \frac{\epsilon^2}{c}$ for some $k \in [n]$ such that $u_i(a_k, \vx^*_{-i}) \leq u_i(a_j, \vx^*_{-i}) - c$. Consider the strategy $\Delta^n \ni \vx_i' = \vx_i^* + x_i^*(a_k) \ve_{j} - x_i^*(a_k) \ve_{k}$. Then, we have
    \begin{align*}
        u_i(\vx_i', \vx^*_{-i}) - u_i(\vx_i^*, \vx^*_{-i}) &= x_i^*(a_k) u_i(a_j, \vx^*_{-i}) - x_i^*(a_k) u_i(a_k, \vx^*_{-i}) \\
        &\geq c x_i^*(a_k) \\
        &> \epsilon^2.
    \end{align*}
    That is, deviating to $\vx_i'$ yields a utility benefit strictly larger than $\epsilon^2$, which contradicts the assumption that $(\vx_i^*, \vx^*_{-i})$ is an $\epsilon^2$-Nash Equilibrium.
\end{proof}

We continue with~\Cref{thm:teamadvtosymme}, which combines~\Cref{lemma:small_z,lemma:closeness} to complete the $\CLS$-hardness reduction of~\Cref{sec:cls}.

\teamadvtosymme*

\begin{proof}
    Since $(\vx^*, \vy^*, \vz^*)$ is an $\epsilon^2$-Nash equilibrium, we have that for any for any deviation $\vy' \in \Delta^n$ of Player $\vy$,
    \begin{align}
        \langle \vx^*, \mat{A} \vy^* \rangle & \leq \langle \vx^*, \mat{A} \vy' \rangle  + \frac{|\Amin|}{\epsilon} \left(\sum_{i = 1}^n z_i (x^*_i - y'_i) + z_{n + i} (y'_i - x^*_i)\right) + \epsilon^2. \label{eq:nash1}
    \end{align}
    Moreover, by considering a deviation of Player $\vx$ again to $\vy'$,
    \begin{align}
        \langle \vx^*, \mat{A} \vy^* \rangle  & \leq \langle \vy', \mat{A} \vy^* \rangle  + \frac{|\Amin|}{\epsilon} \left(\sum_{i = 1}^n z_i (y'_i - y^*_i) + z_{n + i} (y^*_i - y'
        _i)\right) + \epsilon^2 \label{eq:nash2}
    \end{align}
    Adding~\eqref{eq:nash1} and~\eqref{eq:nash2}, and using the fact that $\mat{A}$ is a symmetric matrix,
    \begin{align}
        2 \langle \vx^*, \mat{A} \vy^* \rangle  & \leq \langle \vy', \mat{A} (\vx^* + \vy^*) \rangle + \frac{|\Amin|}{\epsilon} \left(\sum_{i = 1}^n z_i (x^*_i - y^*_i) + z_{n + i} (y^*_i - x^*
        _i) \right) + 2 \epsilon^2 \notag \\
        & \leq \langle \vy', \mat{A} (2 \vy^*) \rangle  + 2\epsilon n |\Amin| + \frac{|\Amin|}{\epsilon} ( 2n\cdot 9 \epsilon \cdot 2\epsilon ) + 2\epsilon^2 \label{eq:bounderror1} \\
        & \leq 2 \langle \vy', \mat{A} \vy^* \rangle + (38 n + 2) |\Amin| \epsilon, \label{eq:final1}
    \end{align}
    where in~\eqref{eq:bounderror1} we use \Cref{lemma:closeness,lemma:small_z}. Also,
    \begin{align}
        \langle \vy^*, \mat{A} \vy^* \rangle & = \langle \vx^*, \mat{A} \vy^* \rangle + \langle \vy^* - \vx^*, \mat{A} \vy^* \rangle \notag \\
        & \leq \langle \vx^*, \mat{A} \vy^* \rangle + 2 \epsilon n |\Amin|. \label{eq:final2}
    \end{align}
    Finally, combining~\eqref{eq:final1} and~\eqref{eq:final2}, we conclude that for any $\vy' \in \Delta^n$,
    \begin{align*}
        \langle \vy^*, \mat{A} \vy^* \rangle \leq \langle \vy', \mat{A} \vy^* \rangle  + (21 n + 1) |\Amin|\epsilon.
    \end{align*}
    This concludes the proof.
\end{proof}

\subsection{Proofs from Section~\ref{sec:symmetric}}
\label{sec:proofs2}

We continue by applying Brouwer's fixed point theorem to show that symmetric min-max optimization problems always admit a symmetric equilibrium.

\begin{lemma}\label{lem:exists}
Let $\calX$ be a convex and compact set. Then, any $L$-smooth, antisymmetric function (\Cref{def:symmetric}) $f:\calX\times \calX \to \R$ admits a symmetric first-order Nash equilibrium $(\vx^*,\vx^*)$.
\end{lemma}

\begin{proof}
We define the function $M: \calX \to \calX$ to be
\begin{equation}
M(\vx') := \prod_{\calX} \left[\vx' - \nabla _{\vx}  f(\vx,\vy)\Big|_{(\vx,\vy)=(\vx',\vx')} \right].
\end{equation}
Given that $f$ is $L$-smooth, we conclude that $M(\vx')$ is $(L+1)$-Lipschitz, hence continuous. Therefore, from Brouwer's fixed point theorem, there exists an $\vx^*$ so that $M(\vx^*) = \vx^*.$ Moreover, the symmetry of $f$ implies that $\nabla_{\vy}f(\vx,\vy) \Big|_{(\vx,\vy)=(\vw,\vw)} =-\nabla_{\vx}f(\vx,\vy) \Big|_{(\vx,\vy)=(\vw,\vw)} $ for all $\vw\in\calX$; as a result,
\begin{equation*}
    \begin{array}{cc}
    \vx^* &= \left[\vx^* - \nabla _{\vx}  f(\vx,\vy)\Big|_{(\vx,\vy)=(\vx^*,\vx^*)} \right] \\
    &= \left[\vx^* + \nabla _{\vy}  f(\vx,\vy)\Big|_{(\vx,\vy)=(\vx^*,\vx^*)} \right].
    \end{array} 
\end{equation*}
Therefore, $(\vx^*,\vx^*)$ is a first-order Nash equilibrium of the symmetric min-max problem with function $f$. 
\end{proof}

We next state a standard lemma that connects first-order optimality with the fixed-point gap of gradient ascent.

\begin{lemma}[\citep{ghadimi2016accelerated}, Lemma 3 for $c=1$]\label{lem:approxsmooth} Let $f(\vx)$ be a $L$-smooth function in $\vx \in \Delta^n$. Define the gradient mapping 
\[G(\vx) := \prod_{\Delta^n}
\left\{\vx+\nabla_{\vx}f(\vx)\right\}-\vx.\]
If $\norm{G(\vx^*)}_2 \leq \epsilon$, that is, $\vx^*$ is an $\epsilon$-approximate fixed point of gradient ascent with stepsize equal to one, then
\[
\max_{\vx^*+\bm{\delta} \in \Delta^n, \norm{\bm{\delta}}_2\leq 1} \bm{\delta}^{\top}\nabla_{\vx} f(\vx^*) \leq \epsilon(L+1).
\]
\end{lemma}

Similarly, the next lemma makes such a connection for min-max optimization problems with coupled constraints; it is mostly extracted from~\citet[Section B.2]{DSZ21}.

\begin{lemma}\label{lem:safe} Let $f(\vx,\vy)$ be a $G$-Lipschitz, $L$-smooth function defined in some polytope domain $\mathcal{D} \subseteq \Delta^n \times \Delta^n$ of diameter $D$. Define the mapping 
\[G(\vx,\vy) := \prod_{\mathcal{D}}
\left\{\vx-\nabla_{\vx}f(\vx,\vy),\vy+\nabla_{\vy}f(\vx,\vy)\right\}-(\vx,\vy).\]
If $\norm{G(\vx^*,\vy^*)}_2 \leq \epsilon$, that is, $(\vx^*,\vy^*)$ is an $\epsilon$-approximate fixed point of (the safe version) of GDA with stepsize equal to one, then 
\[
\langle \vx-\vx^*,\nabla_{\vx}f(\vx^*,\vy^*)\rangle
\geq  -\sqrt{\epsilon}K \textrm { for }\vx \in \mathcal{D}(\vy^*) \textrm{ and }\langle \vy-\vy^*,\nabla_{\vy}f(\vx^*,\vy^*)\rangle
\leq  \sqrt{\epsilon}K \textrm { for }\vy \in \mathcal{D}(\vx^*),
\]
where $\mathcal{D}(\vx^*) = \{\vy:(\vx^*,\vy)\in\mathcal{D}\}$, 
$\mathcal{D}(\vy^*) = \{\vx:(\vx,\vy^*)\in\mathcal{D}\}$ and $K = (L+1)\sqrt{(G+4\sqrt{2})}.$
\end{lemma}
\begin{proof}
Let 
$(\vx_{\Delta},\vy_{\Delta}) = (\vx^*-\nabla_{\vx}f(\vx^*,\vy^*),\vy^*+\nabla_{\vy}f(\vx^*,\vy^*)).$ 
In~\citet[Claim B.2]{DSZ21}, it was shown that for all $(\vx,\vy)\in \mathcal{D}$, we have
\[\langle(\vx_{\Delta},\vy_{\Delta})-(\vx^*,\vy^*),(\vx,\vy)-(\vx^*,\vy^*)\rangle\leq (G+2D)\epsilon.\]
Using the above inequality, it was concluded that  $(\vx^*,\vy^*)$ is an approximate fixed point of the ``unsafe'' version of GDA; specifically,
\[
\left\|\vx^* - \prod_{\mathcal{D}(\vy^*)}\{\vx^* - \nabla_{\vx}f(\vx^*,\vy^*)\}\right\|
\leq \sqrt{(G+2D)\epsilon}
\]
and
\[
\left\|\vy^* - \prod_{\mathcal{D}(\vx^*)}\{\vy^* + \nabla_{\vy}f(\vx^*,\vy^*)\}\right\|
\leq \sqrt{(G+2D)\epsilon}.
\]
We now use~\Cref{lem:approxsmooth} for both inequalities above, together the fact that $D = 2\sqrt{2}$, to conclude that
\[
\langle \vx-\vx^*,\nabla_{\vx}f(\vx^*,\vy^*)\rangle
\geq  -\sqrt{(G+4\sqrt{2})\epsilon}(L+1) \textrm { for }\vx \in \mathcal{D}(\vy^*),\]  and \[\langle \vy-\vy^*,\nabla_{\vy}f(\vx^*,\vy^*)\rangle
\leq  \sqrt{(G+4\sqrt{2})\epsilon}(L+1) \textrm { for }\vy \in \mathcal{D}(\vx^*).
\]
\end{proof}

\subsection{Proofs from Section~\ref{sec:nonsymmetric}}
\label{sec:proofs3}

\input{text/antisymproof}

\subsection{Proofs from Section~\ref{sec:teamzero}}
\label{sec:proofs4}

We conclude with the missing proofs from~\Cref{sec:teamzero}.

\begin{proof}[Proof of~\Cref{lemma:close-teams}]
    For the sake of contradiction, suppose that $\|\vx^* - \vy^*\|_\infty > 2 \epsilon$. Without loss of generality, let us further assume that there is some coordinate $i$ such that $x_i^* - y_i^* > 2 \epsilon$. The payoff difference for Player $\hvz$ when playing the $i$th action compared to action $a_{2n + 1}$ reads
    \begin{align*}
        u(\vx^*, \vy^*, \vz^*, \hat{\vx}^*, \hat{\vy}^*, a_i) - u(\vx^*, \vy^*, \vz^*, \hat{\vx}^*, \hat{\vy}^*, a_{2n + 1}) & = \frac{|\Amin|}{\epsilon} \cdot (x^*_i - y^*_i) - |\Amin| \\
        & > |\Amin|.
    \end{align*}
    By~\Cref{lemma:mass_on_small_reward}, it follows that $\hvz_{2n+1}^* < \frac{\epsilon^2}{|\mat{A}_{\min}|} \leq \epsilon^2$. Moreover, given that $(\vx^*, \vy^*, \vz^*, \hvx^*, \hvy^*, \hvz^*)$ is an $\epsilon^2$-Nash equilibrium, we have
\begin{align*}
        u(\vx^*, \vy^*, \vz^*, \hat{\vx}^*, \hat{\vy}^*, \hat{\vz}^*) & \geq u(\vx^*, \vy^*, \vz^*, \hat{\vx}^*, \hat{\vy}^*, a_i) - \epsilon^2 \\
        & = \langle \vx^*, \mat{A} \vy^* \rangle  - \langle \hvx^*, \mat{A} \hvy^* \rangle + \langle \vx^*, \mat{C} \hvx^* \rangle  + \frac{|\mat{A}_{\min}|}{\epsilon}(x^*_i - y^*_i) - \delta(\hvx, \hvy, \vz) - \epsilon^2 \\
        & \geq \Amin - \langle \hvx^*, \mat{A} \hvy^* \rangle + \langle \vx^*, \mat{C} \hvx^* \rangle + 2 |\mat{A}_{\min}| - \delta(\hvx, \hvy, \vz) - \epsilon^2 \\
        & = |\mat{A}_{\min}| - \langle \hvx^*, \mat{A} \hvy^* \rangle  + \langle \vx^*, \mat{C} \hvx^* \rangle   - \delta(\hvx, \hvy, \vz) - \epsilon^2.
\end{align*}
Now, considering the deviation of Player $\vy$ to $\vy' \defeq \vx^*$,
\begin{align*}
        u(\vx^*, \vy', \vz^*, \hvx^*, \hvy^*, \hvz^*) - u(\vx^*, \vy^*, \vz^*, \hvx^*, \hvy^*, \hvz^*) & \leq \langle \vx^*, \mat{A} \vx^* \rangle + \epsilon^2 |\Amin| - |\Amin| + \epsilon^2 \\
        &\leq -2 + 2 \epsilon^2 \\
        &< - \epsilon^2,
    \end{align*}
    which contradicts the fact that $(\vx^*, \vy^*, \vz^*, \hat{\vx}^*, \hat{\vy}^*, \hat{\vz}^*)$ is an $\epsilon^2$-Nash equilibrium. We conclude that $\norm{\vx^* - \vy^*}_{\infty} \leq 2 \epsilon$; the proof for the fact that $\norm{\hat{\vx}^* - \hat{\vy}^*}_{\infty} \leq 2\epsilon$ follows similarly.
\end{proof}

\begin{proof}[Proof of~\Cref{lemma:smallzteam}]
    We will prove that $\vz_j \leq 9 \epsilon$ for all $j \in [2n]$; the corresponding claim for Player $\hvz$ follows similarly. Fix $i \in [n]$. \Cref{lemma:close-teams} shows that $|y_i^* - x_i^* | \leq 2 \epsilon$. We shall consider two cases.

    First, suppose that $|y_i^* - x_i^*| \leq \nicefrac{\epsilon}{2}$. Then,
    \begin{align*}
        u(\vx^*, \vy^*, \vz^*, \hat{\vx}^*, \hat{\vy}^*, a_{2n + 1}) - u(\vx^*, \vy^*, \vz^*, \hat{\vx}^*, \hat{\vy}^*, a_i) &\geq |\mat{A}_{\min}| - \frac{|\mat{A}_{\min}|}{\epsilon} \cdot (x^*_i - y^*_i)\\
        & \geq \frac{|\mat{A}_{\min}|}{2} \geq \frac{1}{2}
    \end{align*}
    By~\Cref{lemma:mass_on_small_reward}, it follows that $\hat{z}_i \leq 2 \epsilon^2$, and similar reasoning yields $\hat{z}_{n + i} \leq 2 \epsilon^2$. On the other hand, suppose that $|y_i^* - x_i^*| > \nicefrac{\epsilon}{2}$. Without loss of generality, we can assume that $y_i^* - x_i^* \geq 0$; the contrary case is symmetric. Since $\vx^* \in \Delta^n$ and $\vy^* \in \Delta^n$, there is some coordinate $j \in [n]$ such that $y^*_j - x^*_j < 0$. As before, by \Cref{lemma:mass_on_small_reward}, it follows that $\hat{z}^*_i \leq \epsilon^2$ and $\hat{z}^*_{n + j} \leq \epsilon^2$. Now, we consider the deviation
    \begin{align*}
        \Delta^n \ni \vy' = \vy^* + (x_i^* - y_i^*) \ve_i + (y_i^* - x^*_i) \ve_j.
    \end{align*}
    Then, we have
    \begin{align*}
        &u(\vx^*, \vy', \vz^*, \hat{\vx}^*, \hat{\vy}^*, \hat{\vz}^*) - u(\vx^*, \vy^*, \vz^*, \hat{\vx}^*, \hat{\vy}^*, \hat{\vz}^*)  \\ 
        &=  \langle \vx^*, \mat{A} (\vy' - \vy^*) \rangle + \frac{|\Amin|}{\epsilon} \left(\hat{z}^*_i (x^*_i - y'_i - (x^*_i - y^*_i)) + \hat{z}^*_j(x^*_j - y'_j - (x^*_j - y^*_j)) \right) \\
        &\phantom{===============} + \frac{|\mat{A}_{\min}|}{\epsilon} \left(\hat{z}^*_{n+i} ( y'_{i} -x^*_{i} - (y^*_{i} - x^*_{i})) + \hat{z}^*_{n+j}(y'_{j} - x^*_{j} - (y^*_{j} - x^*_{j})) \right) \\
        &= \langle \vx^*, \mat{A} (\vy' - \vy^*) \rangle + \frac{|\Amin|}{\epsilon} \left(\hat{z}^*_i (y_i^* - y'_i) + \hat{z}^*_j (y^*_j - y'_j) + \hat{z}^*_{n + i} (y'_{i} - y^*_{i})  + \hat{z}^*_{n + j} (y'_{ j} - y^*_{j})\right) \\
        &\leq 4\epsilon |\Amin| + \frac{|\Amin|}{\epsilon}\left(\epsilon^2 \cdot 2\epsilon + \hat{z}^*_{n + i} \cdot \left(- \frac{\epsilon}{2} \right) + \epsilon^2 \cdot 2\epsilon \right) \\
        &\leq -\left(\frac{1}{2}\hat{z}^*_{n +i} - 4 \epsilon^2 - 4\epsilon\right) |\Amin|.
    \end{align*}
    At the same time, since $(\vx^*, \vy^*, \vz^*, \hat{\vx}^*, \hat{\vy}^*, \hat{\vz}^*)$ is an $\epsilon^2$-Nash equilibrium, we have
    \begin{align*}
        u(\vx^*, \vy', \vz^*, \hat{\vx}^*, \hat{\vy}^*, \hat{\vz}^*) - u(\vx^*, \vy^*, \vz^*, \hat{\vx}^*, \hat{\vy}^*, \hat{\vz}^*) \geq - \epsilon^2.
    \end{align*}
    Thus,
    \begin{equation*}
        \left(\frac{1}{2}\hat{z}^*_{n +i} - 4 \epsilon^2 -4\epsilon\right) |\mat{A}_{\min}| \leq \epsilon^2 \implies \hat{z}^*_{n +i} \leq 9 \epsilon.
    \end{equation*}
    We conclude that $\hat{z}^*_i \leq \epsilon^2$ and $\hat{z}^*_{n + i} \leq 9\epsilon$. The case where $x_i^* - y_i^* \geq 0$ can be treated similarly.
\end{proof}

\begin{proof}[Proof of~\Cref{theorem:team-hard}]
    Suppose that $(\vx^*, \vy^*, \vz^*, \hat{\vx}^*, \hat{\vy}^*, \hat{\vz}^*)$ is an $\epsilon^2$-Nash equilibrium. We have that for any $\vy' \in \Delta^n$,
    \begin{equation}
        \label{eq:dev1}
        \langle \vx^*, \mat{A} \vy^* \rangle \leq \langle \vx^*, \mat{A} \vy' \rangle + \frac{|\Amin|}{\epsilon} \left(\sum_{i = 1}^n \hat{z}^*_i (y^*_i - y'_i) + \hat{z}^*_{n + i} (y'_i - y^*_i)\right) + \epsilon^2.
    \end{equation}
    Moreover, considering a deviation from $\vx^*$ to $\vy'$,
    \begin{equation}
        \label{eq:dev2}
        \langle \vx^*, \mat{A} \vy^* \rangle + \langle \vx^*, \mat{C} \hvx^* \rangle \leq \langle \vy', \mat{A} \vy^* \rangle + \langle \vy', \mat{C} \hvx^* \rangle + \frac{|\Amin|}{\epsilon} \left(\sum_{i = 1}^n \hat{z}^*_i (y'_i - x^*_i) + \hat{z}^*_{n + i} (x^*_i - y'
        _i)\right) + \epsilon^2.
    \end{equation}
    Summing~\eqref{eq:dev1} and~\eqref{eq:dev2},
    \begin{align}
        2 \langle \vx^*, \mat{A} \vy^* \rangle + \langle \vx^*, \mat{C} \hvx^* \rangle & \leq \langle \vy', \mat{A} (\vx^* + \vy^*) \rangle + \langle \vy', \mat{C} \hvx^* \rangle + \frac{|\Amin|}{\epsilon} \left(\sum_{i = 1}^n \hat{z}^*_i (y^*_i - x^*_i) + \hat{z}^*_{n + i} (x^*_i - y^*
        _i)\right) + 2\epsilon^2 \notag \\
        & \leq 2 \langle \vx^*, \mat{A} \vy' \rangle + \langle \vy', \mat{C} \hvx^* \rangle + 2 \epsilon n |\Amin| + \frac{|\Amin|}{\epsilon} ( 2n \cdot 9 \epsilon \cdot 2\epsilon) + 2\epsilon^2 \notag \\
        & \leq 2 \langle \vx^*, \mat{A} \vy' \rangle + \langle \vy', \mat{C} \hvx^* \rangle + (38n + 2) |\Amin| \epsilon \label{eq:team_game_VI_sum}.
    \end{align}
    Moreover, 
    \begin{align}
        \langle \vx^*, \mat{A} \vx^* \rangle =  \langle \vx^*, \mat{A} \vy^* \rangle + \langle \vx^*, \mat{A} (\vx^* - \vy^*) \rangle \leq \langle \vx^*, \mat{A} \vy^* \rangle + 2 |\Amin| n \epsilon \label{eq:team_game_show_symmetric_nash}.
    \end{align}
    Combining~\eqref{eq:team_game_VI_sum} and~\eqref{eq:team_game_show_symmetric_nash}, we get that for all $\vy' \in \Delta^n$,
    \begin{align}
        \langle \vx^*, \mat{A} \vx^* \rangle + \frac{1}{2} \langle \vx^*, \mat{C} \hat{\vx}^* \rangle \leq \langle \vy', \mat{A} \vx^* \rangle + \frac{1}{2} \langle \vy', \mat{C} \hat{\vx}^* \rangle + (21n + 1 ) |\mat{A}_{\min}|\epsilon. \label{eq:team_game_VI1}
    \end{align}
    Similarly, we can show that for all $\hat{\vx}' \in \Delta^n$,
    \begin{align}
        \langle -\hat{\vx}^*, \mat{A} \hat{\vx}^* \rangle + \frac{1}{2} \langle \vx^*, \mat{C} \hat{\vx}^* \rangle \geq - \langle \hat{\vx}', \mat{A} \hat{\vx}^* \rangle + \frac{1}{2} \langle \vx^*, \mat{C} \hat{\vx}' \rangle - (21n + 1) |\Amin| \epsilon. \label{eq:team_game_VI2}
    \end{align}
    Taking $\vy' = \hat{\vx}'$ in~\eqref{eq:team_game_VI1} and summing with~\eqref{eq:team_game_VI2}, we get that for all $\hat{\vx}' \in \Delta^n$,
\begin{equation*}
        \langle \hat{\vx}', \mat{A} \vx^* \rangle + 
\langle \hat{\vx}', \mat{A} \hat{\vx}^* \rangle + \frac{1}{2} (\langle \hat{\vx}', \mat{C} \hat{\vx}^* \rangle +  \langle \hat{\vx}', \mat{C} \vx^* \rangle) + (42n + 2) |\mat{A}_{\min}| \epsilon \geq \langle \vx^*, \mat{A} \vx^* \rangle + \langle \hat{\vx}^*, \mat{A} \hat{\vx}^* \rangle,
\end{equation*}
    where we used the fact that $\mat{C}$ is skew-symmetric. Now, using the fact that the Nash equilibrium is symmetric, so that $\vx^* = \hvx^*$, we have
\begin{align}
        \langle \hat{\vx}', \mat{A} \vx^* \rangle + \frac{1}{2} \langle \hat{\vx}', \mat{C} \vx^* \rangle + (21n + 1) |\Amin| \epsilon &\geq \langle \vx^*, \mat{A} \vx^* \rangle \notag \\
        &\geq \langle \vx^*, \mat{A} \vx^* \rangle + \frac{1}{2} \langle \vx^*, \mat{C} \vx^* \rangle. \label{eq:team_game_final_eq3}
    \end{align}
    Setting $\mat{A} \defeq - \frac{1}{2} (\mat{R} + \mat{R}^\top)$ and $\mat{C} \defeq \mat{R}^\top - \mat{R}$, \eqref{eq:team_game_final_eq3} shows that 
    \begin{equation*}
        \langle \hvx', \mat{R} \vx^* \rangle \leq \langle \vx^*, \mat{R} \vx^* \rangle + (21n + 1) |\Amin| \epsilon
    \end{equation*}
    for any $\hvx' \in \Delta^n$; \emph{ergo}, $(\vx^*, \vx^*)$ is a symmetric $(21n + 1)  |\Amin| \epsilon$-Nash equilibrium of the symmetric (two-player) game $(\mat{R}, \mat{R}^\top)$, and the proof follows from~\Cref{theorem:PPAD_for _symmetric}.
\end{proof}



%% file: text/antisymproof.tex
In this subsection, our goal is to establish~\Cref{theorem:symmetric-new} that was claimed earlier in the main body, which forms the basis for~\Cref{theorem:uniqueATG,theorem:non-symmetric}.

We consider the symmetric game $(\lineA, \lineA^{\top})$, where $\lineA \in \mathbb{R}^{n \times n}$ is a symmetric matrix. In particular, since $\lineA = \lineA^{\top}$, the two players in the game share the same payoff matrix. The payoff matrix $\lineA$ is constructed based on an underlying graph $G = ([n], E) $ and a parameter $\delta \in (0, 1)$ as follows:
\begin{align*}
    \lineA_{i, j} = 
    \begin{cases}
        \delta & \text{if $i = j$}, \\
        1 & \text{if $(i, j) \in E$}, \\
        0 & \text{otherwise.}
    \end{cases}
\end{align*}

\begin{lemma} \label{lemma:value_of_approx_nash_on_clique}
    Let $(\hat{\vx}, \hat{\vx})$ be an $\epsilon$-well-supported NE of the game where $\hat{\vx}$ is supported on a max clique of size $k$, denoted as $C_k$, then the value $u(\hat{\vx}, \hat{\vx})$ is at least $ 1 - \frac{1}{k} + \frac{\delta}{k} - (\frac{k - \delta}{1 - \delta})  \epsilon$ and $\norm{\hat{\vx} - \vx^*}_{\infty} \leq \frac{k - \delta}{1 - \delta} \epsilon$ where $\vx^* = \frac{1}{k} \sum_{i \in C_k} \ve_{i}$.
\end{lemma}

\begin{proof}
    Since $\hat{\vx}$ is supported on $C_k$, let $i$ be a coordinate that $\hat{\vx}^*$ puts the least probability mass on; that is, $i \in \argmin_{j \in C_k} \hat{x}_j.$ Considering the utility of playing action $a_i$, we have
    \begin{align*}
        u(\ve_i, \hat{\vx}) = \hat{x}_i \cdot \delta + (1 - \hat{x}_i) \cdot 1 = 1 - \hat{x}_i + \hat{x}_i \delta.
    \end{align*}
    Moreover, let $j \in C_k$ be a coordinate such that $\hat{\vx}_j \geq \frac{1}{k}$. It should hold that 
    \begin{align*}
    u(\ve_j, \hat{\vx}) = 1 - (1-\delta)\hat{x}_j \leq 1 - \frac{1}{k} + \frac{\delta}{k}.
    \end{align*}
    Since $(\hat{\vx}, \hat{\vx})$ is an $\epsilon$-well-supported NE, we have
    \begin{align*}
    & u(\ve_j, \hat{\vx}) \geq u(\ve_i, \hat{\vx}) - \epsilon \\
    \Rightarrow  \quad  &1 - \frac{1}{k} + \frac{\delta}{k} + \epsilon \geq 1 - \hat{x}_i + \hat{x}_i \delta \\
        \Rightarrow \quad &\hat{x}_i \geq \frac{1}{k} - \frac{\epsilon}{1 - \delta}.
    \end{align*}
    Thus, for all coordinates $i\in C_k$ , we have $\frac{1}{k} -\frac{\epsilon}{1 - \delta}\leq \hat{x}_i \leq \frac{1}{k} + \frac{\epsilon}{1 - \delta} \cdot (k - 1).$
    It then holds that
    \begin{align*}
        u(\hat{\vx}, \hat{\vx}) & \geq u(\ve_i, \hat{\vx}) - \epsilon \\
        & \geq \left(\frac{1}{k} + \frac{\epsilon}{1 - \delta} \cdot (k - 1)\right) \cdot \delta + \left(1 - \frac{1}{k} - \frac{\epsilon}{1 - \delta} \cdot (k - 1)\right) - \epsilon \\
        & \geq 1 - \frac{1}{k} + \frac{\delta}{k} - \left(\frac{k - \delta}{1 - \delta}\right) \epsilon.
    \end{align*}
\end{proof}

\begin{assumption} \label{assump:parameter}
    For the rest of this subsection, we set the parameters as follows:
    \begin{itemize}
        \item $n \geq k \geq 10;$
        \item $\epsilon < \delta(1 - \delta)/6n^7;$
        \item $\delta \defeq 1/2.$
    \end{itemize}
\end{assumption}

(Using the symbolic value of $\delta$ is more convenient in our derivations below.)

\begin{proof}[Proof of~\Cref{lemma:well_supported_nash_value}]
    We let $\vx^* = \frac{1}{k} \sum_{i \in C_k} \ve_{i}$. The proof considers the following three cases:
    \begin{itemize}
        \item The symmetric $\epsilon$-well-supported NE $(\hat{\vx}, \hat{\vx})$ has support size less than $k$.
        \item The symmetric $\epsilon$-well-supported NE $(\hat{\vx}, \hat{\vx})$ has support size greater than $k$.
        \item The symmetric $\epsilon$-well-supported NE $(\hat{\vx}, \hat{\vx})$ has support size equal to $k$ but is not supported on a clique.
    \end{itemize}
    We proceed to show that for any of the three cases above, we would not be able to have a symmetric $\epsilon$-well-supported NE that achieves value greater than $1 - \frac{1}{k} + \frac{\delta}{k} - \frac{2\delta}{ n^2k^4} + 2 \epsilon$, which is a contradiction.
    \begin{itemize}
        \item For the first case, since the support has size less than $k$, we can find a coordinate $i \in [n]$ such that $\hat{\vx}_i \geq \frac{1}{k - 1}.$ Therefore, the value of playing that action is 
        \begin{align*}
            u(\ve_i, \hat{\vx}) \leq 1 - \frac{1}{k - 1} + \frac{\delta}{k - 1}.
        \end{align*}  
        Since $(\hat{\vx}, \hat{\vx})$ is a symmetric $\epsilon$-well-supported NE, we have
        \begin{align*}
            u(\hat{\vx}, \hat{\vx}) & \leq u(\ve_i, \hat{\vx}) + \epsilon \\
            & \leq 1 - \frac{1}{k - 1} + \frac{\delta}{k - 1} + \epsilon \\
            & \leq 1 - \frac{1}{k} +\frac{\delta}{k} - \frac{2\delta}{n^2k^4} + 2\epsilon,
        \end{align*}
        where in the last step we use~\Cref{assump:parameter}.
        \item For the second case, suppose $|\text{supp}(\hat{\vx})| = m > k$. By \Cref{lemma:not_connected_set}, since the maximum clique size is $k$, we can find a set $\mathcal{S} \subseteq \text{supp}(\hat{\vx})$ with at least $m - k + 1$ elements such that for each coordinate $i \in \mathcal{S}$, we can find a coordinate $j \in \mathcal{S}$ such that $\lineA_{i, j} = \lineA_{j, i} = 0.$ Now we consider the utility of playing action $a_i$ and $a_j$, we have
    \begin{align*}
        & u(\ve_i, \hat{\vx}) = \sum_{l \in \text{supp}(\hat{\vx}) - \{i, j\}} \hat{x}_l \cdot \lineA_{i, l} + \hat{x}_i \cdot \delta + \hat{x}_j  \cdot \lineA_{i, j}, \\
        &u(\ve_j, \hat{\vx}) = \sum_{l \in \text{supp}(\hat{\vx}) - \{i, j\}} \hat{x}_l \cdot \lineA_{j, l} + \hat{x}_j \cdot \delta + \hat{x}_i  \cdot \lineA_{i, j}.
    \end{align*}
    Since $(\hat{\vx}, \hat{\vx})$ is a symmetric $\epsilon$-well-supported NE, we have
    \begin{align*}
        & \sum_{l \in \text{supp}(\hat{\vx}) - \{i, j\}} \hat{x}_l \cdot \lineA_{i, l} + \hat{x}_i \cdot \delta + \hat{x}_j  \cdot\lineA_{i, j} \geq \sum_{l \in \text{supp}(\hat{\vx}) - \{i, j\}} \hat{x}_l \cdot \lineA_{j, l} + \hat{x}_j \cdot \delta + \hat{x}_i  \cdot \lineA_{i, j} - \epsilon\\
        \Rightarrow \quad & \sum_{l \in \text{supp}(\hat{\vx}) - \{i, j\}} \hat{x}_l \cdot \lineA_{i, l}  - \sum_{l \in \text{supp}(\hat{\vx}) - \{i, j\}} \hat{x}_l \cdot \lineA_{j, l}   \geq \hat{x}_j \cdot \delta - \hat{x}_i \cdot \delta - \epsilon.
    \end{align*}
    Now, by moving all the probability mass from action $j$ to action $i$, we form a new strategy $\vx' = \hat{\vx} + \hat{x}_j \cdot \ve_i - \hat{x}_j \cdot \ve_j $ such that
    \begin{align}
        u(\vx', \vx') - u(\hat{\vx}, \hat{\vx}) &=  2 \hat{x}_j\left(\sum_{l \in \text{supp}(\hat{\vx}) - \{i ,j\}} \hat{x}_l \cdot \lineA_{i, l} - \sum_{l \in \text{supp}(\hat{\vx}) - \{i ,j\}} \hat{x}_ml\cdot \lineA_{j, l}\right) + 2 \hat{x}_i \cdot \hat{x}_j \cdot \delta \notag \\
        & \geq  2 \hat{x}_j \cdot ((\hat{x}_j - \hat{x}_i) \delta - \epsilon)+ 2 \hat{x}_i \cdot\hat{x}_j \cdot \delta \notag \\
        & \geq 2 {\hat{x}_j}^2 \delta - 2 \epsilon \label{eq:marginal_gain_when_one_move}. 
    \end{align}
        Suppose there is a coordinate $i \in \mathcal{S}$ such that $\hat{x}_i \geq \frac{1}{nk^2}$ then from \eqref{eq:marginal_gain_when_one_move}, we have 
        \begin{align*}
            u(\hat{\vx}, \hat{\vx}) &\leq u(\vx^*, \vx^*) - 2\cdot(\frac{1}{nk^2})^2 \cdot \delta + 2 \epsilon\\
            & = 1 - \frac{1}{k} + \frac{\delta}{k} - \frac{2\delta}{ n^2k^4} + 2 \epsilon.
        \end{align*}
        If $\hat{x}_i < \frac{1}{nk^2}$ for any $i \in \mathcal{S}$, then there exists a coordinate $l \not \in \mathcal{S}$ such that $\hat{x}_l > \frac{1 - (m - k + 1)\cdot\frac{1}{nk^2}}{m - (m - k + 1)} \geq \frac{1}{k} + \frac{1}{k^2}.$ Then, considering the utility when playing action $a_l$,
        \begin{align}
            u(\ve_l, \hat{\vx}) &\leq 1 \cdot (1 - \frac{1}{k} - \frac{1}{k^2}) + \delta \cdot (\frac{1}{k} + \frac{1}{k^2}) \notag \\
            & = 1 - \frac{1}{k} + \frac{\delta}{k} - \frac{1}{k^2} + \frac{\delta}{k^2} \notag \\
            &< 1 - \frac{1}{k} + \frac{\delta}{k} - \frac{2\delta}{n^2k^4}, \label{eq:case_two_vaue_of_nash}
        \end{align}
        where in~\eqref{eq:case_two_vaue_of_nash} we used~\Cref{assump:parameter}.
        
        Since the $l$th action is played with positive probability and $(\hat{\vx}, \hat{\vx})$ is an $\epsilon$-well-supported NE, we have $u(\hat{\vx}, \hat{\vx}) \leq u(\ve_l, \hat{\vx}) + \epsilon < 1 - \frac{1}{k} + \frac{\delta}{k} - \frac{2\delta}{n^2k^4} + 2\epsilon.$ 
        \item For the third case, since the support is not on a clique, the exists at least coordinates $i, j$ such that $\hat{x}_i >0, \hat{x}_j>0,$ and $\lineA_{i, j} = \lineA_{j, i} = 0.$ Similarly as case two, if $\hat{x}_i \geq \frac{1}{nk^2}$ or $\hat{x}_j \geq \frac{1}{nk^2}$, then we have  
        \begin{align*}
            u(\hat{\vx}, \hat{\vx}) &\leq u(\vx^*, \vx^*) - 2\cdot(\frac{1}{nk^2})^2 \cdot \delta + 2 \epsilon\\
            & = 1 - \frac{1}{k} + \frac{\delta}{k} - \frac{2\delta}{ n^2k^4} + 2\epsilon.
        \end{align*}
        If $\hat{x}_i < \frac{1}{nk^2}$ and $\hat{x}_j < \frac{1}{nk^2}$, then there exists an coordinate $l$ such that $\hat{x}_l \geq \frac{1 - 2 \cdot \frac{1}{nk^2}}{k - 2} > \frac{1}{k} + \frac{1}{k^2}.$ Same as \eqref{eq:case_two_vaue_of_nash}, we conclude that $u(\hat{\vx}, \hat{\vx})$ is at most $1 - \frac{1}{k} + \frac{\delta}{k} - \frac{2\delta}{n^2k^4} + 2\epsilon.$
    \end{itemize}
    The proof is complete.
\end{proof}

We now construct a new symmetric identical payoff game $(\mat{B}, \mat{B})$, where $\mat{B}$ is defined as
    \begin{align}
        \mat{B} = \begin{bmatrix}
         \lineA_{1, 1} & \cdots & \lineA_{1, n}  & r\\
         \vdots & \ddots & \vdots & \vdots \\
         \lineA_{n, 1} & \cdots & \lineA_{n,n} & r\\
         r & \cdots & r  & V\\
    \end{bmatrix}. \label{eq:unique_NP_matrix}
    \end{align}
Above, $V \defeq 1 - \frac{1}{k} + \frac{\delta}{k}$ and $r \defeq 1 - \frac{1}{k} + \frac{\delta}{k} - \frac{\delta}{n^2k^4} + 3\epsilon$. (We caution that the values of $V$ and $r$ have been set differently compared to~\Cref{sec:nonsymmetric}.) Similarly to our derivation in~\cref{eq:three_exact_NE}, it follows that the symmetric (exact) Nash equilibria of this game can only be in one of the following forms:
\begin{enumerate}
        \item $(\vx^*, \vx^*)$ with $\vx^* \defeq \ve_{n+ 1}$; \label{eq:NE_case1}
        \item $(\vx^*, \vx^*)$ with $\vx^* \defeq \frac{1}{k} \sum_{i \in C_k} \ve_i$, where $C_k \subseteq [n]$ is a clique in $G$ of size $k$ \label{eq:NE_case2};
        \item $(\vx^*, \vx^*)$ with $\vx^* \defeq \frac{1}{2} \ve_{n + 1} + \frac{1}{2k} \sum_{i \in C_k} \ve_i$, where $C_k \subseteq [n]$ is a clique in $G$ of size $k$ \label{eq:NE_case3}.
\end{enumerate}

We now show the following lemma.

\begin{lemma} \label{lemma:NP_completeness_for_two_symmetric_well_supported}
    For any $\epsilon$-well-supported NE $(\hat{\vx}, \hat{\vx})$ in game $(\mat{B}, \mat{B})$, it holds that $\norm{\hat{\vx} - \vx^*}_{\infty} \leq 2n^6\epsilon$, where $(\vx^*, \vx^*)$ is an exact NE in one of the three cases above (\ref{eq:NE_case1},\ref{eq:NE_case2},\ref{eq:NE_case3}).
\end{lemma}

\begin{proof}
    First, we observe that since $V > r$, clearly $(\ve_{n +1}, \ve_{n + 1})$ is a $\epsilon$-well-supported NE; in this case, $\norm{\vx' - \vx^*}_{\infty} = 0.$ Furthermore, since $r = 1 - \frac{1}{k} + \frac{\delta}{k} - \frac{\delta}{n^2k^4} + 3 \epsilon$, the game does not attain any $\epsilon$-NE with value less than $1 - \frac{1}{k} + \frac{\delta}{k} - \frac{\delta}{n^2k^4} + 2 \epsilon.$ Suppose the game admits an symmetric $\epsilon$-well-supported NE $(\hat{\vx}, \hat{\vx})$ where $\hat{\vx}$ is supported only on the first $n$ actions. Since $u(\hat{\vx}, \hat{\vx}) \geq 1 - \frac{1}{k} + \frac{\delta}{k} - \frac{\delta}{n^2k^4} + 2 \epsilon > 1 - \frac{1}{k} + \frac{\delta}{k} - \frac{2\delta}{n^2k^4} + 2\epsilon,$ taking  $\vx^*$ as in \eqref{eq:NE_case2}, we conclude that $\norm{\hat{\vx} - \vx^*}_{\infty} \leq \frac{k - 1}{1 - \delta} \epsilon < 2n^6\epsilon$ from \cref{lemma:well_supported_nash_value}. 
    
     We proceed to the case where there is a mixed symmetric $\epsilon$-well-supported Nash $(\hat{\vx}, \hat{\vx})$ between the last action and the rest of actions such that $0 < \hat{x}_{n+1} < 1.$ Denote $\hat{x}_{n + 1} = \alpha$ and $\eta(\cdot)$ to denote the renormalization operation. Since $\left(\eta(\hat{\vx}_{[1 \cdots n]}), \eta(\hat{\vx}_{[1 \cdots n]})\right)$ is a symmetric strategy profile, from~\citet[Proposition 4]{MCLENNAN2010683}, we conclude that there is at least one coordinate $i \in \text{supp}(\hat{\vx}) - \{n+1\}$ such that
     \begin{align*}
         u_A\left(\ve_i, \eta(\hat{\vx}_{[1 \cdots n]})\right) \leq 1 - \frac{1}{k} + \frac{\delta}{k} = V,
     \end{align*}
     where $u_A$ is the utility when the payoff matrix is $\lineA.$ Since $\left(\eta(\hat{\vx}_{[1 \cdots n]}), \eta(\hat{\vx}_{[1 \cdots n]})\right)$ is an $\epsilon$-well-supported NE, we have
     \begin{align}
         & u_B(\ve_i, \hat{\vx}) \geq u_B(\ve_{n + 1}, \hat{\vx}) - \epsilon \notag \\
         \Rightarrow \quad & (1 - \alpha)\cdot V + \alpha \cdot r \geq (1 - \alpha)r + \alpha V - \epsilon \notag \\
         \Rightarrow \quad & \alpha \leq \frac{1}{2} + \frac{\epsilon}{2(V - r)}, \label{eq:lower_bound_on_alpha}
     \end{align}
     where $u_B$ is the utility function when the payoff matrix is $\mat{B}$. Plugging in the value of $V$ and $r$ and using~\Cref{assump:parameter}, we find that $\alpha \leq \frac{1}{2} + 2n^6\epsilon \leq \frac{2}{3}$.
     
     Now, we observe that for any action in the support other than the last action $a_i$, the utility of playing such action $u_B(a_i, \hat{\vx}) = u_A(a_i, \hat{\vx}_{[1 \cdots n]}) + r \alpha$. Since $(\hat{\vx}, \hat{\vx})$ is an $\epsilon$-well-supported Nash Equilibrium, we have $u_B(a_i, \hat{\vx}) \geq \max_{j} u_B(a_j, \hat{\vx}) - \epsilon$ for all pairs $(i, j) \in \text{supp}(\hat{\vx}).$ Since $\alpha \leq \frac{2}{3},$ it follows that $u_A\left(a_i, \eta(\hat{\vx}_{[1 \cdots n]})\right) \geq \max_j u_A\left(a_j, \eta(\hat{\vx}_{[1 \cdots n]})\right) - 3\epsilon$ for any pairs $(i, j) \in \text{supp}(\hat{\vx}) - \{n + 1\}.$ Thus, we conclude that $\left(\eta(\hat{\vx}_{[1 \cdots n]}), \eta(\hat{\vx}_{[1 \cdots n]})\right)$ forms a symmetric $3\epsilon$-well-supported Nash Equilibrium in game $(\lineA, \lineA)$. Further, the value of playing the last action is $(1 - \alpha)r + \alpha V > r$, and so the only situation where there is a mixed Nash between the last action and the rest actions is when $u_A\left(\left(\eta(\hat{\vx}_{[1\cdots n]}), \eta(\hat{\vx}_{[1\cdots n]})\right)\right) \geq r - \epsilon$. Therefore, by \Cref{lemma:value_of_approx_nash_on_clique} and \Cref{lemma:well_supported_nash_value}, we conclude that $\hat{\vx}_{[1 \cdots n]}$ is supported on a clique of size $k.$
     There exits at least one coordinate $i \in [n]$, with $0 < \hat{x}_i$ and $ \hat{x}_i \geq \frac{1}{k},$ such that
    \begin{align*}
        u_A(\ve_{i}, \eta(\hat{\vx}_{[1 \cdots n]})) \geq 1 - \frac{1}{k} + \frac{\delta}{k} = V.
    \end{align*}
    Since $(\hat{\vx}, \hat{\vx})$ is an $\epsilon$-well-supported Nash, we have $u(\ve_i, \hat{\vx}) \leq u(\ve_{n +1}, 
    \hat{\vx}) + \epsilon$, and so this gives
    \begin{align}
        & (1 - \alpha) \cdot V + \alpha \cdot r \leq (1 - \alpha) \cdot r + \alpha \cdot V + \epsilon \label{eq:alpha_upper_lower_bound}\\
        \Rightarrow \quad &  \alpha \geq \frac{1}{2} - \frac{\epsilon}{2(V - r)}. \label{eq:range_alpha}
    \end{align}
    Using~\Cref{assump:parameter} and 
    combining with~\eqref{eq:lower_bound_on_alpha},
    \begin{align*}
        \frac{1}{2} - 2n^6\epsilon \leq \alpha \leq \frac{1}{2} + 2n^6\epsilon.
    \end{align*}
    By taking $\vx^*$ as in \eqref{eq:NE_case3}, we conclude that $\norm{\hat{\vx} - \vx^*}_{\infty} \leq 2n^6 \epsilon$.
\end{proof}

\begin{theorem}\label{thm:betadefined}
    For any $\epsilon$-NE $(\vx, \vx)$ in game $(\mat{B}, \mat{B})$, it holds that $\norm{\vx - \vx^*}_{\infty} \leq n^6 \sqrt{\epsilon}$, where $(\vx^*, \vx^*)$ is an exact NE in one of the forms specified above (\ref{eq:NE_case1},\ref{eq:NE_case2},\ref{eq:NE_case3}).
\end{theorem}

\begin{proof}
    \citet[Lemma 3.2]{Chen09:Settling} showed that from any $\epsilon^2/8$-NE $(\vx, \vy)$ in any two player bimatrix game, one can construct (in polynomial time) an $\epsilon$-well-supported NE $(\vx', \vy')$ such that $\norm{\vx - \vx'}_{\infty} \leq \frac{\epsilon}{4}$ and $\norm{\vy - \vy'}_{\infty} \leq \frac{\epsilon}{4}$. By setting $\epsilon' \defeq \frac{\epsilon^2}{8}$ for the $\epsilon$ defined in ~\Cref{lemma:NP_completeness_for_two_symmetric_well_supported} the proof follows.
\end{proof}

The proof of \Cref{theorem:symmetric-new} follows directly by observing that having two symmetric $\epsilon$-NE $(\vx, \vx)$ and $(\vy, \vy)$ such that $\norm{\vx - \vy}_{\infty} > 2n^6 \sqrt{\epsilon}$ would imply that the game $(\mat{B}, \mat{B})$ has two distinct exact NE, which in turn implies that there is a clique of size $k$ in the graph.

We finally conclude this subsection by stating and proving an auxiliary lemma that was used earlier.

\begin{lemma} \label{lemma:not_connected_set}
    For any graph $G = (V, E)$ with $n$ vertices, if the maximum clique has size $k$, then we can form a set $\mathcal{S} \subseteq V$ of size at least $n - k + 1$ such that for any vertex $i \in \mathcal{S}$, there exists a vertex $j \in \mathcal{S}$ such that $i$ and $j$ are not connected.
\end{lemma}

\begin{proof}
    Suppose the largest set $\mathcal{S}$ we can form has cardinality $|\mathcal{S}| < n - k + 1$, this implies there is a set $\mathcal{S}' = V - \mathcal{S}$ with at least $n - (n - k ) = k$ vertices such that each vertex in $\mathcal{S'}$ is connected to all other vertices in $G$. However if this is the case, there is at least one vertex $v \not \in \mathcal{S'}$ that are connected to all vertices in $\mathcal{S'}$. This contradicts the fact that maximum clique has size $k$.
\end{proof}